\documentclass[10pt,journal,compsoc]{IEEEtran}

\usepackage{amsmath,amssymb,amsfonts,amsthm}
\usepackage{algorithmic}
\usepackage{algorithm}
\usepackage{subfig}
\usepackage{url}
\usepackage{graphicx}
\usepackage{flushend}
\usepackage{xcolor}
\usepackage{ragged2e}
\usepackage{caption}
\ifCLASSOPTIONcompsoc
  \usepackage{cite}
\else
  \usepackage{cite}
\fi

\ifCLASSINFOpdf
\else
\fi

\graphicspath{ {./figures/} }

\def\BibTeX{{\rm B\kern-.05em{\sc i\kern-.025em b}\kern-.08em
    T\kern-.1667em\lower.7ex\hbox{E}\kern-.125emX}}

\floatname{algorithm}{Mechanism}
\newtheorem{theorem}{Theorem}
\newtheorem{lemma}{Lemma}
\newtheorem{example}{Example}
\newtheorem{definition}{Definition}

\newtheorem{remark}{Remark}

\begin{document}

\title{Differentially Private Community Detection\\
  in $h$-uniform Hypergraphs}

\author{Javad~Zahedi~Moghaddam, Aria~Nosratinia
}


\IEEEtitleabstractindextext{
\begin{abstract}
\justifying This paper studies the exact recovery threshold subject to preserving the privacy of connections in $h$-uniform hypergraphs. Privacy is characterized by the $(\epsilon, \delta)$-hyperedge differential privacy (DP), an extension of the notion of $(\epsilon, \delta)$-edge DP in the literature. The hypergraph observations are modeled through a $h$-uniform stochastic block model ($h$-HSBM) in the dense regime. We investigate three differentially private mechanisms: stability-based, sampling-based, and perturbation-based mechanisms. We calculate the exact recovery threshold for each mechanism and study the contraction of the exact recovery region due to the privacy budget, $(\epsilon, \delta)$. Sampling-based mechanisms and randomized response mechanisms guarantee pure $\epsilon$-hyperedge DP where $\delta=0$, while the stability-based mechanisms cannot achieve this level of privacy. The dependence of the limits of the privacy budget on the parameters of the $h$-uniform hypergraph is studied. More precisely, it is proven rigorously that the minimum privacy budget scales logarithmically with the ratio between the density of in-cluster hyperedges and the cross-cluster hyperedges for stability-based and Bayesian sampling-based mechanisms, while this budget depends only on the size of the hypergraph for the randomized response mechanism.

\end{abstract}   

}

\maketitle

\IEEEdisplaynontitleabstractindextext

\IEEEpeerreviewmaketitle

\IEEEraisesectionheading{\section{Introduction}\label{sec:introduction}}

\IEEEPARstart{P}rivacy-preserving learning and inference on sensitive graph datasets has been widely acknowledged as an important area of inquiry~\cite{8661646,9521831,9669059,9964113,10049709,10232888,10304338,10877783,11078900}.  
 One of the learning tasks is  community detection~\cite{10947335}  which  aims to uncover latent group structures in networks by analyzing the connectivity patterns among nodes.  These connectivity patterns can be represented simply by an edge between two nodes in a graph, as in the case of a friendship network, or by a hyperedge in hypergraphs, as in the case of higher-order relations like transactions involving multiple bank accounts, individuals, and institutions in financial networks. A simple graph and hypergraph are shown in  Fig.~\ref{fig:graph_hypgraph}.  Publishing the result of a community detection algorithm can inadvertently reveal some affiliations or behavioral traits in a network that are deemed private~\cite{9057414,10236973}. Therefore, data curators seek privacy-preserving methods for reliable community recovery. Differential privacy (DP)~\cite{dwork2006differential,dwork2008differential,dwork2014algorithmic} has emerged as a framework that characterizes meaningful privacy guarantees. In the context of edge differential privacy~\cite{mulle2015privacy}, the data curator allows accurate community detection while preserving edge information privacy.

Unsurprisingly, the addition of privacy mechanisms affects the performance of the community detection algorithm. Therefore, it is valuable to understand how differential privacy impacts the information-theoretic limits of community detection. Fundamental limits reveal the conditions under which no algorithm, regardless of computation, can reliably recover the underlying community structures. 

This problem is at the forefront of current research in statistical learning~\cite{mohamed2022differentially,seif2023differentially}, especially for the exact recovery wherein the community detection algorithm is tasked with retrieving the underlying labels of all members with error probability approaching zero in the asymptote of large graph size. In~\cite{mohamed2022differentially,seif2023differentially}, it was rigorously proved that edge-DP restricts the exact recovery fundamental limits compared with non-private cases in graphs. 

This study rigorously investigates the impact of differential privacy on the exact recovery fundamental limits of $h$-uniform hypergraphs. The hypergraph privacy problem is distinct from the graph privacy in ways described in the sequel, and not been studied thus far.

\begin{figure}
\centering
\subfloat[]{\includegraphics[width=0.24\textwidth]{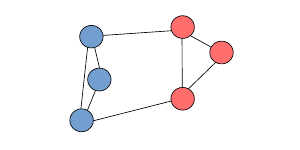}}
\hfil
\subfloat[]{\includegraphics[width=0.24\textwidth]{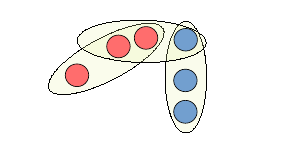}} 
\caption{  Underlying communities in red and blue for (a) common graph, (b) hypergraph with 3-uniform hyperedges }
\label{fig:graph_hypgraph}
\end{figure}

\subsection{Contribution}
\label{subsec:contributions}
Inspired by the edge-DP notion~\cite{mulle2015privacy}, we begin by defining hyperedge differential privacy (hyperedge-DP) for $h$-uniform hypergraphs. Then, we utilize the generative $h$-uniform stochastic block model ($h$-HSBM) in the dense regime to represent hypergraph observations, where formation of hyperedges is governed by a probability distribution that depends on the community memberships of the vertices they connect, such that in-cluster hyperedges occur with higher probability than cross-cluster hyperedges. Within this structure, we then 
 adapt the non-private optimal estimator of a $h$-HSBM into a $(\epsilon, \delta)$-hyperedge DP estimator. We investigate three distinct differentially private algorithms in this context: stability-based mechanisms, sampling-based mechanisms, and the random response mechanisms. We derive the exact recovery bound for each mechanism within our established framework. 

Our analysis of the fundamental limits for optimal estimators reveals that under $(\epsilon, \delta)$-hyperedge DP, the exact recovery regions contract compared with a non-private scenario. In addition, while sampling-based and perturbation-based mechanisms ensure pure $\epsilon$-hyperedge DP with $\delta=0$, stability-based mechanisms cannot provide this level of privacy.  Moreover, the parameters of the $h$-HSBM determine the most private DP scenario, i.e., the smallest values of $(\epsilon,\delta)$ among all mechanisms. More precisely, the smallest $\epsilon$ for stability-based DP and Bayesian sampling DP mechanisms grows logarithmically with the ratio between the densities of in-cluster and cross-cluster hyperedges, $\log(a/b)$, where $a$ and $b$ are model parameters controlling these densities. This means that as $\log(a/b)$  increases, the distinguishability of communities gets better, but leads to a less private DP scenario. For the randomized response mechanism, the minimum privacy budget $\epsilon$  grows logarithmically with the size of the hypergraph, $\log(n)$, and this means community detection under hyperedge DP is less effective for large-scale networks. Furthermore, the exponential sampling-based mechanism is the only one whose privacy budget is independent of hypergraph parameters and can provide differential privacy for any positive $\epsilon$.

The discovered exact recovery thresholds under hyperedge-DP are provided for each mechanism in Table~\ref{tab:contribution} where $h$ is the order of hyperedges, $\lambda = \Omega_n(1)$ and $t = \Omega_n(1)$.

\begin{table*}[ht]

    \centering
    \caption{  Information-theoretic limits of differentially private exact recovery for different mechanisms under hyperedge-level differential privacy }
    \begin{tabular}{c|c|c|c}
       Scenario & \( \delta\) & \( \epsilon \) & Threshold \\ \hline
         Non-private &  0 & 0 &  $(\sqrt{a} - \sqrt{b})^2  > 2^{h-1} $    \\
         &  &  & \\
         RR DP & 0 &\( \Omega_n(\log(n))\) &    $(\sqrt{a+\lambda}-\sqrt{b+\lambda})^2  > 2^{h-1}$ \\
         &  &  & \\
         Stability-based DP & \(n ^ {-t}\) & \( \geq \frac{t+1}{2}  \log(\frac{a}{b}) \) &  \( a+b-\sqrt{ \frac{(t+1)^2}{4 \epsilon^2}  \big(\frac {h}{h-1} \big)^{2h-2}+4ab}  > 2^{h-1} \)  \\
         &  &  &\\
         Exponential Sampling DP& 0& \( \geq \Omega_n(1) \)&    $\epsilon(a-b)  > 2^{h-1} $ \\
         &  &  & \\
         Bayesian Sampling DP & 0& \( \geq \epsilon_0 =  \log( \frac{a}{b})\) &  \(  (1-e^{-\epsilon_0})(a - b)  > 2^{h-1} \) 
    \end{tabular} 
    \label{tab:contribution}
\end{table*}

\subsection{Related Works}
\subsubsection{Exact recovery thresholds in graphs} 
Discovering the information-theoretic limits of exact recovery under various scenarios, and developing efficient algorithms that achieve these fundamental limits are an active area of research in statistical learning and information theory. Identifying these thresholds requires a probabilistic framework that incorporates distinct connectivity probabilities for edges within and between underlying communities. To this end, the stochastic block model (SBM)~\cite{holland1983stochastic} is the most celebrated generative model in the literature. Research on SBMs has demonstrated that exact recovery is feasible only in dense regimes, where the connectivity parameters of communities scale logarithmically with the number of nodes~\cite{abbe2015exact}. Numerous studies have addressed the problem of finding fundamental limits of the exact recovery, and exploited several computationally effective algorithms including spectral methods\cite{yun2014accurate,chen2015spectral,mossel2015consistency,8793181,yun2019optimal,gangrade2019efficient,gaudio2022exact,zhang2023fundamental}, semi-definite programming (SDP)\cite{abbe2015exact,massoulie2014community,hajek2016achieving,hajek2016achieving_extensions,jalali2016exploiting,esm_2021_sdp,yan2021covariate,cbm_javad,javadTnse24}, and approximate belief propagation followed by a majority voting procedure \cite{sandon2017community,Hajek_Wu_Xu_2018,saad_side18,wu2021streaming} to achieve these thresholds in graphs. Comprehensive surveys covering the subject in random graphs can be found in the works~\cite{abbe2017community,ning2023comprehensive}. 
\subsubsection{Exact recovery thresholds in $h$-uniform hypergraphs} In the general hypergraphs, finding the exact recovery thresholds poses many challenges primarily because of the variation in hyperedge orders. However, significant advances have been achieved in addressing specific scenarios like uniform hypergraphs, where all hyperedges have the same order. Inspired by the standard SBM, the $h$-uniform hypergraph SBM ($h$-HSBM) was introduced to analyze uniform hypergraphs\cite{ghoshdastidar2014consistency}. In $h$-HSBM, each hyperedge connects precisely $h$ nodes, and its occurrence probability depends on the communities to which these $h$ nodes belong. Subsequent researchers have further explored the recovery limits of $h$-HSBMs by extending the techniques previously employed in common SBMs analysis \cite{ghoshdastidar2017consistency,kim2017community,lin2017fundamental,kim2018stochastic,chien2018community,ahn2018hypergraph,ahn2019community,chien2019minimax,cole2020exact,lee2020robust,zhang2022exact,gaudio2023community,wang2023projected,deng2023strong}.

\subsubsection{Community detection under differential privacy }

Research studies on private community detection are more focused on proposing edge-DP algorithms. In \cite{mulle2015privacy}, the Privacy Integrated Graph Clustering (PIG) was proposed to ensure edge-DP by perturbing the edges of the original graph. Nguyen \emph{et al.} \cite{nguyen2016detecting} adopted the Louvain algorithm as their backend and introduced LouvainDP for input perturbation along with ModDivisive for algorithm perturbation, which utilizes the exponential sampling mechanism to ensure edge-DP. LDPGen was presented in \cite{qin2017generating}, which is a novel multi-phase technique that incrementally clusters users based on their connections to different partitions of the population to mitigate excessive noise injection in privacy mechanisms while preserving crucial graph properties. Ji \emph{et al.} \cite{8999786} developed Differentially Private Community Detection (DPCD), and they rigorously showed that DPCD guarantees both edge-DP and attribute-DP. In~\cite{hehir2021consistency}, they showcased how employing the edge flipped mechanism to make spectral methods edge-DP impacts the convergence rate of the algorithm in contrast to non-private scenarios. 
Recently, DPRec has been introduced in~\cite{ZHOU2026103621} that integrates privacy preservation and recommendation services together, and it also guarantees a sub-linear convergence rate during the model optimization process in theory, which can also satisfy 
$\epsilon$-differential privacy.
 
 Performance bounds for algorithms providing DP community detection were studied in~\cite{mulle2015privacy,nguyen2016detecting,qin2017generating,imola2021locally,ji2019differentially,hehir2021consistency,ZHOU2026103621}. Performance bounds for {\em optimal} edge-DP community detection, however, have only recently been studied rigorously, and only for common graphs ~\cite{mohamed2022differentially,seif2023differentially}. For common graphs, exploiting stability-based, sampling-based, and perturbation mechanisms, these works showed that differential privacy shrinks the fundamental recovery regions of maximum likelihood compared with non-private cases. They also developed some limits on privacy budgets in different scenarios~\cite{mohamed2022differentially,seif2023differentially}. In this paper, we address the problem of hyperedge privacy in the context of $h$-uniform hypergraphs and study how differential privacy constraints impact the fundamental limits of exact recovery of maximum likelihood in comparison with non-private scenarios under different edge-DP mechanisms.

\section{Problem Settings}
\label{sec:Model and Settings}

\subsection{$h$-HSBM and Exact recovery}

A hypergraph $H$ on the set of $n$ nodes $\mathcal{V} = \{1,\ldots,n\}$ is defined as a collection of hyperedges, each connecting or joining multiple nodes. The order of a hyperedge refers to the number of nodes it connects. A hypergraph is called $h$-uniform if all its hyperedges have the same order $h$. The set of all $h$-subsets of $\mathcal{V}$, namely $\mathcal{W} = {\binom{\mathcal{V}}{h}}$, defines all possible $h$-order hyperedges over the set of nodes $\mathcal V$. A hypergraph $H$ has a collection of hyperedges that is denoted $\mathcal{W}_H$ which is a subset of $\mathcal{W}$.

We now define a generative model of random hypergraphs. Let $0 < q\leq p < 1$, and assume we draw a random $h$-uniform hypergraph $\mathcal{H}$ in the following manner:
\begin{itemize}
\item A vector of binary node labels, also known as communities and denoted $\boldsymbol \sigma \in \{\pm 1\}^n$ is drawn uniformly and independently.
\item Each $w = \{v_1,\cdots,v_h\} \in \mathcal{W}$ is a hyperedge in $\mathcal{W}_\mathcal{H}$ with probability:
\begin{align}
\label{eq:inc_crc_prb}
Pr(w \in \mathcal{W}_\mathcal{H}) = \begin{cases} p &\text{if $\sigma_{v_1}=\sigma_{v_2}=\cdots=\sigma_{v_h}$} \\
q & \text{otherwise}. \end{cases}
\end{align}
\item Any two distinct $w, w'$ occur independently.
\end{itemize}
The random model thus generating $\mathcal H$ is denoted $h$-HSBM$(n,p,q)$, and an outcome of this model is shown with $H$. The hyperedge $w \in \mathcal{W}_\mathcal{H}$ is called \emph{in-cluster} hyperedge with respect to (w.r.t) $\boldsymbol \sigma$ if it occurs with probability $p$. Otherwise, it is a \emph{cross-cluster} hyperedge w.r.t. $\boldsymbol \sigma$. Community detection algorithms aim to find communities to a desirable level of recovery from the observation of an instance (realization) $H$ of the hypergraph. In this paper, we consider exact recovery, which is defined as follows.
\begin{definition}
\label{def:exact_recovery}
For a hypergraph $\mathcal{H}$ generated by ${h}$-$HSBM(n,p,q)$ on the ground truth communities $\boldsymbol\sigma^*$, exact recovery is possible if
\[
\exists \quad \widehat{ \boldsymbol \sigma}(\cdot): \quad Pr(\widehat{\boldsymbol \sigma}({\mathcal H}) \not\in \{\boldsymbol\sigma^\ast,-\boldsymbol\sigma^\ast \}) = o_n(1)
\]
Exact recovery on $\mathcal{H}$ is impossible if:
\[
\forall \quad \widehat{ \boldsymbol \sigma}(\cdot): Pr(\widehat{\boldsymbol \sigma}(\mathcal{H}) \not\in \{\boldsymbol\sigma^\ast,-\boldsymbol\sigma^\ast \}) = 1- o_n(1)
\]
\end{definition}
Exact recovery is feasible only in connected and dense regimes. Therefore, the average degree of nodes in random hypergraph $\mathcal{H}$ must be at least $\frac{c(h-1)\log n}{\binom{n-1}{h-1}}$ for some $c > 1$ \cite{kim2018stochastic} with probability converging to one. Hence, the connectivity parameters of a ${h}$-$HSBM(n,p,q)$ must belong to following regimes:
\[
p = \frac{a\log n}{\binom{n-1}{h-1}} \;, \qquad q = \frac{b \log n}{\binom{n-1}{h-1}}
\]
for some positive constants $a$ and $b$. We focus on assortative random hypergraphs where $a > b$.

\subsection{$(\epsilon, \delta)$-hyperedge differential privacy}
Inspired by neighborings graph and $(\epsilon, \delta)$-edge DP notions~\cite{mulle2015privacy} , we first define neighbor hypergraphs and $(\epsilon, \delta)$-hyperedge DP. 
\begin{definition} [Neighbor Hypergraphs]
\label{def:neighbor_hypergraph_sigma}
Two hypergraphs ${{H}}$ and ${\tilde H}$ are called neighbors if they differ only in one hyperedge, i.e.:
\[
 |\mathcal{W}_H \Delta \mathcal{W}_{{\tilde H}}| = 1
\]
where $\Delta$ is the symmetric set difference operator.
\end{definition}
\begin{definition} [$(\epsilon, \delta)$-hyperedge DP] 
\label{def:hypedgeDP}
For some $\epsilon \in \mathbb{R}^{+}$ and $\delta \in (0, 1]$, an estimator $\hat{\boldsymbol{\sigma}}$ guarantees $(\epsilon, \delta)$-hyperedge DP if for two given neighbor hypergraphs $H$ and ${\tilde H}$:
\begin{align}
\label{eq:hypedgeDP}
Pr(\hat{\boldsymbol{\sigma}}({{H}}) = \boldsymbol{\sigma}) \leq e^{\epsilon}  \; Pr(\hat{\boldsymbol{\sigma}}({\mathcal{H} ={\tilde H}}) = \boldsymbol {\sigma}) + \delta.
\end{align}
Similar to $(\epsilon, \delta)$-edge DP, we call the case of $\delta = 0$ as pure $\epsilon$-hyperedge DP.
\end{definition} 
In the rest of the paper, $\hat{\boldsymbol{\sigma}}(H)$ implicitly shows $H$ is a realization of a random hypergraph $\mathcal{H} $. 

\section{Diferentially Private Community detection on $h$-uniform Hypergraphs}
\label{sec:Main Results}
This section investigates the impact of differential privacy mechanisms on community detection, particularly analyzing the exact recovery thresholds of the maximum likelihood estimator for $h$-uniform hypergraphs in differentially private scenarios. To this end, we study the recovery threshold of community detection subject to three privacy mechanisms discussed in the literature: stability-based mechanisms, randomized response mechanisms, and sampling-based mechanisms, as outlined in the following subsections.

\subsection{Stability-based Mechanism}
A stable estimator satisfies the hyperedge-DP property if an adversary who knows the communities cannot exactly determine whether a specific hyperedge exists in hypergraph $H$. Distance to instability is the core notion in the stability-based mechanism, which we define as follows for a uniform hypergraph with respect to labeling.

\begin{definition}[Distance to instability of $H$ w.r.t. $\hat{\boldsymbol{\sigma}}$]
\label{def:dist2inst}
The distance to instability of $H$ w.r.t. $\hat{\boldsymbol{\sigma}}$ is defined as:
\begin{align}
\label{eq:dist2inst}
d({H};{\hat{\boldsymbol{\sigma}}}) = \{\min_k:  \exists {\tilde H}, |\mathcal{W}_H \Delta \mathcal{W}_{{\tilde H}}| \leq k, \hat{\boldsymbol{\sigma}}(H)\neq \hat{\boldsymbol{\sigma}}({\tilde H})\}. 
\end{align}
\end{definition}

Inspired by Propose-Test-Release~\cite{dwork2014algorithmic}, our stability-based mechanism starts from a {\em non-private} estimator to recover communities from hypergraph $H$, $\hat{\boldsymbol{\sigma}}({H})$. Then, it calculates the distance to the instability of the non-private $\hat{\boldsymbol{\sigma}}({H})$ based on Definition~\ref{def:dist2inst}. The mechanism then compares the stability of the private estimator on the graph $H$ with a privacy threshold, which depends on $(\epsilon, \delta)$ budget. Finally, the mechanism outputs the non-private estimate $\hat{\boldsymbol{\sigma}}(H)$ if stability is larger than the privacy threshold, otherwise it generates and outputs a random community vector, See {\em Mechanism~\ref{mech:dist}} as an instantiation of Propose-Test-Release below. 

\begin{algorithm}[H]
\small{
\caption{:$\mathcal{M}_{dist}(H;{\hat{\boldsymbol{\sigma}}})$}
\label{mech:dist}
\begin{algorithmic}[1]
 \STATE {\bfseries Input:} $H(\mathcal{V}, \mathcal{W}_H) \in \mathcal{H}$
 \STATE {\bfseries Output:} $\mathcal{M}_{dist}(H;{\hat{\boldsymbol{\sigma}}})$ (private community vector)
\IF{$d ({H};{\hat{\boldsymbol{\sigma}}})  + \text{Lap}(1/\epsilon)> \frac{\log{1/\delta}}{\epsilon} $}
\STATE $\mathcal{M}_{dist}(H;{\hat{\boldsymbol{\sigma}}}) \leftarrow \hat{\boldsymbol{\sigma}}$
\ELSE
\STATE $\mathcal{M}_{dist}(H;{\hat{\boldsymbol{\sigma}}}) \leftarrow \perp$ (random community vector) 
\ENDIF
\end{algorithmic}}
\end{algorithm}
\noindent where $\text{Lap}(.)$ is a Laplacian random variable with the parameter shown in the argument.

\begin{lemma}
\label{lem:hypedgedp_dist}
For any community detection algorithm $\hat{\boldsymbol{\sigma}}$, $\mathcal{M}_{dist}(H;\hat{\boldsymbol{\sigma}})$ guarantees $(\epsilon, \delta)$-hyperedge DP. 
\end{lemma}
\begin{proof}
\label{proof:general_stability}
To guarantee the hyperedge-DP, we first need to calculate the sensitivity of $d({H};{\hat{\boldsymbol{\sigma}}})$. Assume a pair of neighbor hypergraphs $H$ and ${\tilde H}$. By triangle inequality, $|d({H};{\hat{\boldsymbol{\sigma}}}) - d({{\tilde H}};{\hat{\boldsymbol{\sigma}}})| \leq 1$, therefore sensitivity is bounded by one and the rest of proof follows from~\cite{dwork2014algorithmic}.
\end{proof}

\begin{theorem}
\label{thm:dist_mle_2} 
For a given $h$-uniform random binary hypergraph $H$ sampled from ${h}$-$HSBM(n,p,q)$ with same-size communities, $\mathcal{M}_{dist}(H;\hat{\boldsymbol{\sigma}}_{\text{ML}})$ guarantees $(\epsilon, \delta)$-hyperedge DP community detection when:
\begin{subequations}
    \label{eq:mle-2comm}
    \begin{align}
    & \frac{2 \epsilon}{t+1}  \geq \log(\frac{a}{b}) \quad \text{and}\label{eq:mle-2comm_a}\\
    & a+b- 2\sqrt{ \frac{(t+1)^2}{16 \epsilon^2}  \big(\frac {h}{h-1} \big)^{2h-2}+ab} \geq 2^{h-1} \label{eq:mle-2comm_b}
    \end{align}
\end{subequations}
where $\delta = n^{-t}$, $p = \frac{a\log n}{\binom{n-1}{h-1}}$ , $ q = \frac{b \log n}{\binom{n-1}{h-1}}$ , and $a \geq b > 0$.
\end{theorem}

\begin{proof}
See Appendix~\ref{prf:thm:dist_mle_2}
\end{proof}

\begin{remark}
\label{rmk:stb1}
The exact recovery region of $\mathcal{M}_{dist}(H;\hat{\boldsymbol{\sigma}}_{\text{ML}})$ for fixed $(\epsilon,\delta)$ is larger in $h$-HSBMs with lower hyperedge order. In the same manner, $h$-HSBM parameters and the exact recovery requirement constrain the available privacy budgets $(\epsilon, \delta)$.
\end{remark}

To elaborate more on Remark~\ref{rmk:stb1}, consider the following Example.

\begin{example}
\label{xmp:stab}
 Let $\delta = n^{-t}$, where $t$ is a positive scalar and:
\[
\mu (a,\epsilon) \triangleq a+1- \sqrt{\frac{(t+1)^2}{4 \epsilon^2}   \big(\frac {h}{h-1} \big)^{2h-2}+4a}
\]
If $b=1$, then $\alpha = a$ and $\mu (a,\epsilon) \geq 2^{h-1}$ determines the exact recovery threshold. Under this setting, green areas in Fig.~\ref{fig:exact_recovery_stability} indicate regions where exact recovery is possible. Fig.~\ref{fig:exact_recovery_stability_sub1} and Fig.~\ref{fig:exact_recovery_stability_sub2} demonstrate the exact recovery region for uniform random hypergraph with hyperedges of order $h=3$ is larger than the uniform random hypergraph with $h=4$ for fixed $\epsilon$ and $\delta$. Additionally,  white  regions in all sub-figures of Fig.~\ref{fig:exact_recovery_stability} illustrate where assortative mixing exceeds the limit for a fixed $\epsilon$ and $\delta$ due to~\eqref{eq:mle-2comm_a}, and the algorithm cannot achieve exact recovery under hyperedge-DP.
\end{example}

\begin{figure*}[htbp]
    \centering
    \subfloat[$h=3$, $t=1$, $b=1$]{\includegraphics[width=0.24\textwidth]{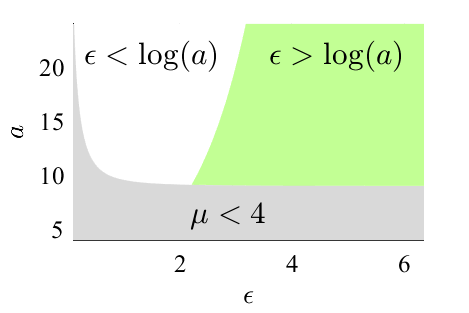}
    \label{fig:exact_recovery_stability_sub1}}
    \hfill
    \subfloat[$h=4$, $t=1$, $b=1$]{\includegraphics[width=0.24\textwidth]{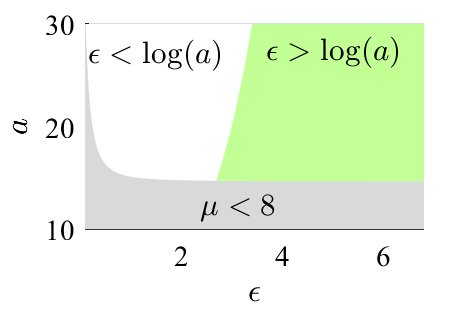}
    \label{fig:exact_recovery_stability_sub2}}
    \hfill
    \subfloat[$h=3$, $t=3$, $b=1$]{\includegraphics[width=0.24\textwidth]{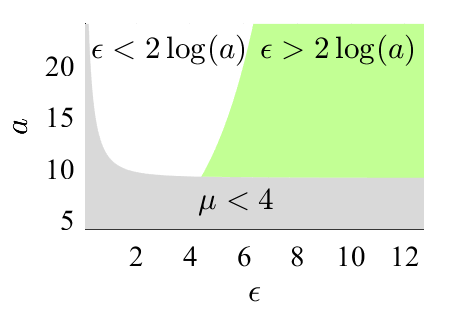}
    \vspace{-10pt}
    \label{fig:exact_recovery_stability_sub3}}
    \hfill
    \subfloat[$h=4$, $t=3$, $b=1$]{\includegraphics[width=0.24\textwidth]{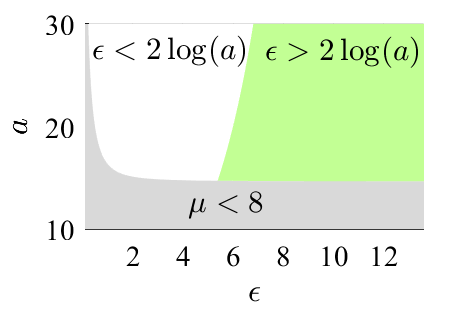}
    \label{fig:exact_recovery_stability_sub4}}
    
    \caption{ Exact recovery regions for Example~\ref{xmp:stab}. Gray: irrecoverable ($\mu < 2^{h-1}$); White: only non-private recovery  ($\mu > 2^{h-1}$, $\epsilon < \frac{t+1}{2} \log{\frac{a}{b}}$). Green: both non-private recovery and private recovery under hyperedge DP constraint ($\mu > 2^{h-1}$, $\epsilon \geq \frac{t+1}{2} \log{\frac{a}{b}}$).  }
    \label{fig:exact_recovery_stability}
\end{figure*}

\begin{remark}
For $\delta=n^{-t}, t>0$, the smallest $\epsilon$ under which $\mathcal{M}_{dist}(H;\hat{\boldsymbol{\sigma}}_{\text{ML}})$ retains asymptotic exact recovery is:
\begin{align*}
   \epsilon = \frac{t+1}{2}  \log(\frac{a}{b})
\end{align*}
In this case, the achievability bound is:
\begin{align*}
       a+b- \frac{1}{\log(\frac{a}{b})} \sqrt{ \big(\frac {h}{h-1} \big)^{2h-2}+ 4ab \log^2(\frac{a}{b})} \geq 2^{h-1}
\end{align*}
This shows that the larger the privacy budget, the more limited the achievability bound becomes. Moreover, a larger exponent $t$ leads to a bigger $\epsilon$ and a reduction in privacy. For Example~\ref{xmp:stab}, this is illustrated in Fig.~\ref{fig:exact_recovery_stability}, where higher $t$ (i.e. a lower $\delta$) results in shrinking exact recovery regions.
\end{remark}

\begin{remark}
For sufficiently large $n$, if $\epsilon \to \infty$, the sufficiency guarantee for exact recovery collapses to 
\begin{align*}
    (\sqrt{a} - \sqrt{b} )^2 \geq 2^{h-1}
\end{align*}
which is consistent with the exact recovery bound in\cite{kim2018stochastic} under no privacy protection.     
\end{remark}

\begin{remark}
Since $\sqrt{x+y} \leq \sqrt{x}+\sqrt{y}$ for positive $x$ and $y$, for every $\epsilon$ and $t$ that satisfy~\eqref{eq:mle-2comm_a}, exact recovery is possible if:
\begin{align}
\label{eq:eq:mle-2comm_2}
    (\sqrt{a} - \sqrt{b})^2 \geq 2^{h-1} \bigg[1 + \frac{t+1}{2 \epsilon} \times \big(\frac{h}{2h-2}\big)^{h-1}\bigg]
\end{align}
When $h=2$, this bound coincides with the exact recovery bound for stability mechanisms for common graphs in~\cite{mohamed2022differentially}.
\end{remark}


\subsection{Randomized Response Mechanism}

Randomized response (RR) is a general technique for differential privacy, which is adopted in this section for $(\epsilon, \delta)$-hyperedge DP community detection in $h$-uniform hypergraphs. In our setting, a randomized response mechanism perturbs ${H}$, resulting in $\tilde{H}$ so that $H$ and $\tilde H$ have the same vertices $\mathcal{V}$, and each hyperedge $w \in \mathcal{W}$ is flipped with probability $\mu$, i.e.:
\begin{align}
\label{mu_rr}
Pr( w \in \mathcal{W}_{{\tilde H}} | w \in \mathcal{W}_{H} ) = 1-\nu
\end{align}  
for some fixed $0<\nu<1$. Now, we can prove if $\nu = \frac{1}{e^{\epsilon} + 1}$, then the Randomized Response mechanism, $\mathcal{M}_{\text{RR}}(H)$, guarantees $\epsilon$-hyperedge DP.
\begin{algorithm}
  \caption{: $\mathcal{M}_{\text{RR}}(H)$}
  \label{mech:RR}
 \small{
 \begin{algorithmic}[1]
         \STATE {\bfseries Input:} $H(\mathcal{V}, \mathcal{W}_H) \in \mathcal{H}$
         \STATE {\bfseries Output:} A Labeling vector $\hat{\boldsymbol{\sigma}} \in \Sigma$
         \STATE Perturb $H \rightarrow \tilde{H} $ with probability $\nu = \frac{1}{e^{\epsilon} + 1}$
         \STATE Apply community detection algorithm on $\tilde{H}$ 
         \STATE Output $\hat{\boldsymbol{\sigma}}(\tilde{H})$ 
  \end{algorithmic}}
\end{algorithm}

\begin{theorem} 
\label{thm:rr_mechanism} 
If $\epsilon = \Omega_n \big( \log(n) \big) $ such that $e^{-\epsilon} = \frac{\lambda \log(n)}{\binom{n-1}{h-1}}$, the mechanism $\mathcal{M}_{\operatorname{RR}}(H)$ asymptotically guarantees $\epsilon$-edge DP,  and satisfies exact recovery for binary balanced communities if 
\begin{align}
\label{eq: rr_mechanism}
    (\sqrt{a+\lambda}-\sqrt{b+\lambda})^2 > 2^{h-1}
\end{align}
where $\lambda = \Omega_n(1)$.
\end{theorem}

\begin{proof}
 See~Appendix\ref{prf:them:rr_mechanism}.
\end{proof}

\begin{remark}
    Using $\mathcal{M}_{\operatorname{RR}}(H)$  shrinks the exact recovery regions compared with the non-private scenarios. This happens since, for $a \geq b \geq 0$,
    \[
    \sqrt{a} - \sqrt{b} \geq \sqrt{a+\lambda}-\sqrt{b+\lambda}.
    \]
\end{remark}


\subsection{Sampling-based Mechanisms}
Sampling methods, where a central collector randomly discards collected responses, have been widely employed in the literature to enhance privacy. Within the community detection and differential privacy framework, we examine the influence of two commonly used mechanisms, the Bayesian mechanism and the Exponential mechanism, on the exact recovery bound of $h$-uniform hypergraphs.

\subsubsection{Bayesian Sampling Mechanism}
The posterior probability for the labels is given by:
\begin{align}
\label{eq:posterior_b}
    Pr (\boldsymbol{\sigma} | {H}) = \frac{Pr({H} | \boldsymbol{\sigma}) Pr(\boldsymbol{\sigma})}{Pr({H} )}
\end{align}
The Bayesian sampling mechanism $\mathcal{M}_{Bayes} ({H})$ selects a $\hat{\boldsymbol{\sigma}} (H)$ according to the posterior probability above. 

\begin{algorithm}[H]
\caption{: $\mathcal{M}_{Bayes} ({H})$}
\label{mech:bayes_samp}
\small{
\begin{algorithmic}[1]
 \STATE {\bfseries Input:} $H(\mathcal{V}, \mathcal{W}_H) \in \mathcal{H}$
 \STATE {\bfseries Output:} A community vector $\hat{\boldsymbol{\sigma}} (H) \in \Sigma$
\STATE Sample $\hat{\boldsymbol{\sigma}}$ according to $Pr (\hat{\boldsymbol{\sigma}} | {H})$ in \eqref{eq:posterior_b}
\end{algorithmic}}
\end{algorithm}

\begin{theorem} 
  \label{thm:bayes_samp}
  The Mechanism $\mathcal{M}_{\text{Bayes}} (H)$ guarantees $\epsilon$-hyperedge DP,  $\forall \epsilon \geq \epsilon_0 = \log \big(\frac{a}{b}\big)$, and exact recovery is possible for same-size binary communities if 
\begin{align}
   (1-e^{-\epsilon_0})(a - b) > 2^{h-1}
\end{align} 

 \end{theorem}

\begin{proof}
See Appendix~\ref{prf:thm:bayes_samp}
\end{proof}
\begin{remark}
The Bayesian sampling mechanism needs computation of the posterior distribution of a community vector; thus, it needs knowledge of the system's parameters $(a,b)$. In practice, often this information is not readily available. 
\end{remark}

\subsubsection{Exponential Sampling Mechanism}
To circumvent the practical difficulties with Bayesian sampling, the exponential sampling mechanism replaces the (unknown) posterior distribution with an exponential distribution. 
For binary labels, we define $\Psi(H; \hat{\boldsymbol{\sigma}})$ as the set of cross-cluster hyperedges in $H$ w.r.t. labeling $\hat{\boldsymbol{\sigma}}$.
\begin{algorithm}
  \caption{: $\mathcal{M}_{\text{Expo.}}(H)$}
  \label{mech:expo_samp}
  \small{
  \begin{algorithmic}[1]
    \STATE {\bfseries Input:} $H(\mathcal{V}, \mathcal{W}_H) \in \mathcal{H}$
    \STATE {\bfseries Output:} A community vector $\hat{\boldsymbol{\sigma}} \in \Sigma$. 
    \STATE Sample $\hat{\boldsymbol{\sigma}}$ with probability $e^{-\epsilon |\Psi(H; \boldsymbol{\hat \sigma})| }$
  \end{algorithmic}}
\end{algorithm}

\begin{theorem}
\label{thm:exponential_samp}
  The exponential sampling mechanism  $\mathcal{M}_{\text{Expo.}}(H)$ guarantees $\epsilon$-hyperedge DP, and exact recovery is possible for the same-size binary communities if:
  \begin{align}
     \epsilon(a-b) > 2^{h-1}
 \end{align}
\end{theorem}

\begin{proof}
See Appendix~\ref{prf:thm:exponential_samp}
\end{proof}

\begin{remark}
    In this method, the disparity between in-cluster and cross-cluster hyperedges has a direct effect on the privacy budget. The bigger the disparity, the smaller the $\epsilon$ it allows while maintaining exact recovery. Conversely, in higher privacy scenarios where $\epsilon$ is small, the exact recovery regions will contract.
\end{remark}

The core results of this paper are independent of the time complexity of privacy-preserving methods. However, for completeness, Table~\ref{tab:contribution_comp} compares the time complexity of various techniques. Both exponential and Bayesian sampling DP require sampling over all possible community assignments, which is exponential in the number of nodes. Step~3 of the stability-based mechanism involves finding a nearby hypergraph where the labeling changes, with time complexity $O(n^{\log n})$, due to $O(\log n)$ hyperedge modifications. The $n^h$ term in randomized response arises from constructing the perturbed hypergraph in Mechanism~\ref{mech:RR}. Finally, $O(\phi(n))$ denotes the time complexity of non-private recovery (e.g., $\phi(n) = n^2 \log n$ for the spectral method~\cite{gaudio2023community}).

\begin{table}[ht]
    \centering  
    \caption{ The time complexity of a community detection algorithm under non-private and differentially private mechanisms. }
    \begin{tabular}{ll}
        {Scenario} & {Time Complexity} \\[4pt] \hline\\
        Non‑private                                & $O(\phi(n))$ \\[4pt]
        Randomized Response DP                                & $O\bigl(\phi(n) + n^{h}\bigr)$ \\[4pt]
        Stability‑based DP                         & $O\bigl(\phi(n) + n^{\log(n)}\bigr)$ \\[4pt]
        Exponential Sampling DP                       & $O\bigl(\exp(n)\bigr)$ \\[4pt]
        Bayesian Sampling DP                          & $O\bigl(\exp(n)\bigr)$
    \end{tabular}
    \label{tab:contribution_comp}    
\end{table}


\section{ Experimental Results }
\label{sec:expr}

The experimental results in this section highlight the performance of the proposed private community detection algorithms and the tradeoffs between privacy and community recovery. We test the RR-based performance mechanism on generated graphs (SBMs). 
For each parameter setting $(\epsilon, a, b)$, we perform 100 Monte Carlo simulations. In each iteration, we estimate the community labels using tensor trace maximization (TTM) and normalized hypergraph cut (NHCut), as described in~\cite{ghoshdastidar2017uniform}. To evaluate performance, we compute the normalized Hamming distance between the true and estimated labels (misclassification error), which serves as an estimate of the error probability, i.e. $ \Pr(\hat{\boldsymbol{\sigma}} \neq \boldsymbol{\sigma}^\ast)$.

The first experiment investigates how assortativity, which is the ratio of the in-cluster to cross-cluster hyperedge probabilities, i.e. $\frac{a}{b}$, affects the estimated error probability.  In this experiment, we analyze the misclassification error under both the non-private and RR-based hyperedge DP scenarios. In the non-private scenario, TTM and NHCut algorithms recover labels from the original graph, whereas in the RR-based hyperedge DP scenario, they recover labels from the perturbed graph by Mechanism~\ref {mech:RR}.  

To study the effect of assortativity, we consider a balanced binary 3-uniform hypergraph with 100 nodes, and we set $b =1$ and vary $a$. In addition, we fix the privacy budget $\epsilon = 7$ for the RR mechanism to isolate the effect of assortativity on the misclassification error under the RR Mechanism. This choice of $\epsilon$ guarantees exact recovery for $a > 10.6008$ based on Theoreom~\ref {thm:rr_mechanism} . The simulation results are shown in Fig.\ref{fig:rr_exp:a}, where the gray shaded area indicates regions where exact recovery is impossible, $(\sqrt{a} - \sqrt{b})< 2^{h-1}$, the white region indicates regions where exact recovery is achievable only in the non-private scenario, and the green shaded area indicates regions where exact recovery is achievable for both non-private case and under the hyperedge DP scenario with Mechanism~\ref{mech:RR}. These regions show how much the exact recovery threshold moves w.r.t. $a$ under hyperedge DP with Mechanism ~\ref{mech:RR} for the first experiment, where $n=100$, $b=1$, and $\epsilon = 7$.

The simulation outcomes in Fig.~\ref{fig:rr_exp:a} illustrate that a larger $a$ reduces the probability of error across all scenarios. This observation supports the findings in Theorem~\ref{thm:rr_mechanism}, where larger values of $a$ increase the disparity between the distributions of in-cluster and cross-cluster hyperedges. As a result, the quantity $\sqrt{a+\lambda} - \sqrt{b+\lambda} $ is more likely to exceed the threshold $2^{h-1} = 4$. Note that $\lambda$ is a fixed positive scalar since $\epsilon$ is fixed. 
Experiments indicate that the estimation error for private scenarios surpasses that of non-private scenarios. This disparity arises from the introduction of excessive noise during the perturbation phase in observations.

The second experiment assesses the influence of the privacy budget $\epsilon$ on the estimation error within a random 3-uniform hypergraph. In a hypergraph with $n=100$, we fix the in-cluster and cross-cluster hyperedge probabilities to isolate the impact of privacy budget. We set $a=13$ and $b=1$, which guarantees exact recovery for the non-private scenario. Fig.~\ref{fig:rr_exp:eps} shows that as $\epsilon$ increases, the probability of error decreases. This occurs because larger values of $\epsilon$ reduce the flipping probabilities $\mu = 1/(1+e^\epsilon)$ during the hypergraph perturbation step. Higher $\epsilon$ reduces the noise introduced into the system by the privacy mechanism, and decreases the estimation error.   This can also be derived directly from Theorem~\ref{thm:rr_mechanism}. When $\epsilon \to \infty$, the setting tends to the non-private case. Theorem~\ref{thm:rr_mechanism} indicates that as $\epsilon$ increases, and accordingly $\lambda$ decreases, the quantity $(\sqrt{a+\lambda} - \sqrt{b+\lambda})^2$ increases for fixed $a$ and $b$. Consequently, this increases the likelihood of entering a region where exact recovery is Possible.  The green shaded regions in Fig.~\ref{fig:rr_exp:eps} show where exact recovery is possible based on the choice of $a$. To be more precise, ~\eqref{eq: rr_mechanism} implies that to guarantee exact recovery under hyperedge DP for the choice of $a=13$, the privacy budget should be bigger than $5.8611$.

\begin{figure}[htbp]
    \centering
    \subfloat[ Effect of assortativity on exact recovery threshold. Non-private setting: $b=1, n=100$, and the hyperedge DP setting: $\epsilon=7, b=1, n=100$ ]{\includegraphics[width=0.5\textwidth]{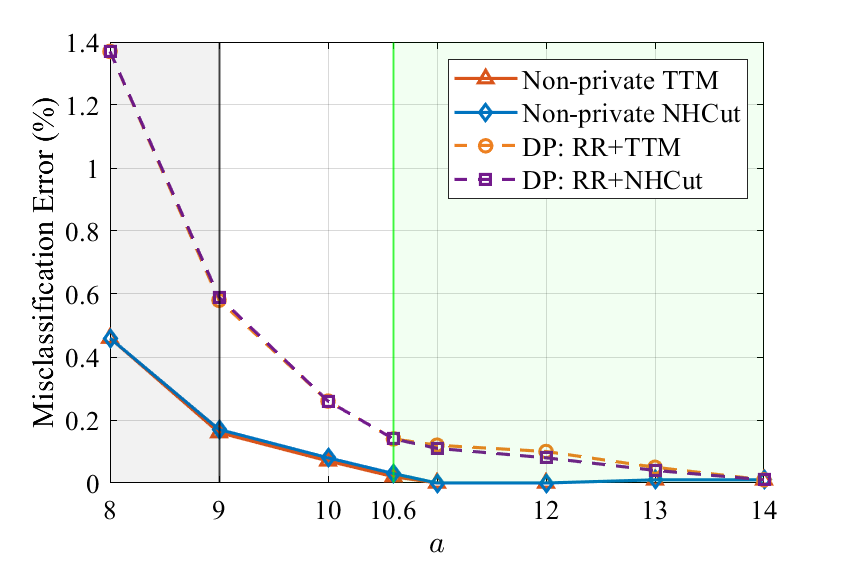}
    \label{fig:rr_exp:a}}
    \hfill
    \subfloat[ Effect of privacy budget on exact recovery threshold under the hyperedge DP for $a=13, b=1, n=100$ ]{\includegraphics[width=0.5\textwidth]{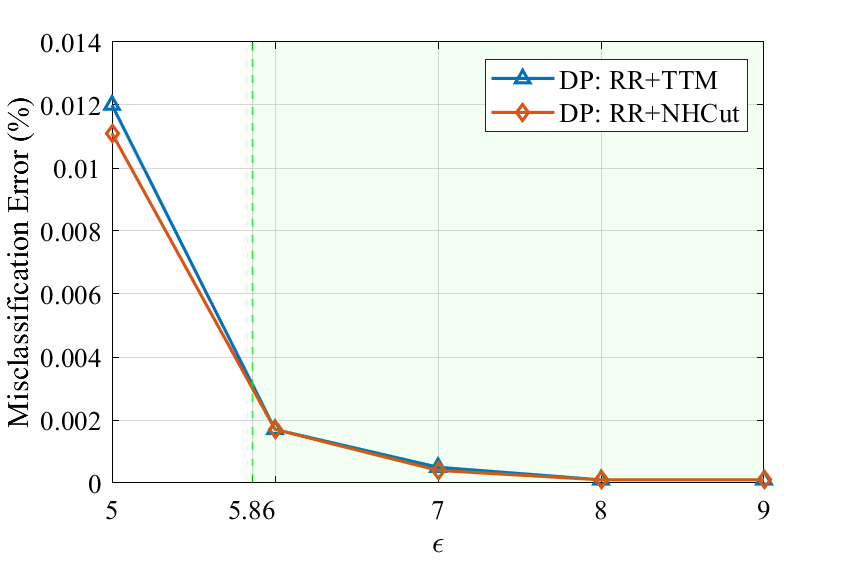}
    \label{fig:rr_exp:eps}}
    
    \caption{  Simulation results for a 3-uniform hypergraph. In (a), the gray region indicates where exact recovery is impossible, the white region indicates where exact recovery is achievable only in the non-private scenario, and the green region indicates where exact recovery is achievable in both the non-private and hyperedge DP scenarios. }
    \label{fig:rr_exps}
\end{figure}
\section{Conclusion}
This paper addresses the problem of identifying the exact recovery thresholds for \(h\)-HSBMs under hyperedge DP. Firstly, it introduces the notion of \((\epsilon, \delta)\)-hyperedge DP, and then it investigates how such privacy constraints can impact the information-theoretic limits of community detection compared to non-private cases. By exploring three differentially private mechanisms, including stability-based, randomized response, and sampling-based mechanisms, we demonstrate that the privacy budget \((\epsilon, \delta)\) causes the maximum likelihood estimator to fail in meeting the exact recovery requirements in regions where recovery is achievable in non-private scenarios. Additionally, we show that all mechanisms can provide \(\epsilon\)-hyperedge DP, with \(\delta = 0\), except for stability-based mechanisms. Furthermore, we prove that there are limitations on the privacy budget \((\epsilon, \delta)\), such as \(\epsilon \geq (0.5 - \log(\delta) / 2 \log(n)) \log(\frac{a}{b})\) for stability-based mechanisms, \(\epsilon \geq \log(\frac{a}{b})\) for sampling mechanisms, and \(\epsilon = \Omega_n { \big( \log (n) \big)}\) for randomized response mechanism. 
%
\appendices

\label{appendix}

\section{Proof of Theorem~\ref{thm:dist_mle_2}}
\label{prf:thm:dist_mle_2}
We define the error event $E$, and the two types of output events of $\mathcal{M}_{dist}$, as follows:
\begin{align*}
E &= \{ \mathcal{M}_{dist}(H;\hat{\boldsymbol{\sigma}}_{\text{ML}})\neq \boldsymbol{\sigma}^\ast \}\\
O &= \{ \mathcal{M}_{dist}(H;\hat{\boldsymbol{\sigma}}_{\text{ML}}) = \hat{\boldsymbol{\sigma}}_{\text{ML}} \} \\
R &= \{ \mathcal{M}_{dist}(H;\hat{\boldsymbol{\sigma}}_{\text{ML}}) = \perp\}
\end{align*}
The error probability of $\mathcal{M}_{dist}$ is: 
\begin{align}
\label{eq:dist_alg_error}
&Pr(E) \nonumber\\
&= Pr (E|R) \cdot Pr(R)+ Pr(E|O)\cdot Pr(O)\nonumber\\
&\le Pr(R)+ Pr(E,O)\nonumber\\
&\le Pr(R)+ Pr(\hat{\boldsymbol{\sigma}}_{\text{ML}} \neq \boldsymbol{\sigma}^\ast)\nonumber\\
& = Pr\Big(d({H};{\hat{\boldsymbol{\sigma}}}) + \text{Lap}(1/\epsilon) < \frac{\log(1/\delta)}{\epsilon} \Big) + Pr(\hat{\boldsymbol{\sigma}}_{\text{ML}} \neq \boldsymbol{\sigma}^\ast) \nonumber \\
& =  Pr\Big(\text{Lap}(1/\epsilon) < \frac{\log(1/\delta)}{\epsilon} - d({H};{\hat{\boldsymbol{\sigma}}})\Big)  + Pr(\hat{\boldsymbol{\sigma}}_{\text{ML}} \neq \boldsymbol{\sigma}^\ast) 
\end{align}
where the randomness comes from $h$-HSBM and $\text{Lap}(1/\epsilon)$. For any positive $t$, let $\delta \triangleq n^{-t}$. Then:
\begin{align}
\label{eq:prob_outputting_null}
&Pr\Big(\text{Lap}(1/\epsilon) < \frac{\log(1/\delta)}{\epsilon} - d({H};{\hat{\boldsymbol{\sigma}}})\Big) \nonumber \\
& \leq Pr\Big(d({H};{\hat{\boldsymbol{\sigma}}}) \geq  \frac{t+1}{\epsilon} \log(n) \Big) \nonumber \\
& \quad \qquad \times Pr\Big(\text{Lap}(1/\epsilon) < \frac{(t+1)\log(n)-\log(\delta)}{\epsilon} \Big) \nonumber \\
& \qquad  + Pr\Big( d({H};{\hat{\boldsymbol{\sigma}}}) <   \frac{t+1}{\epsilon} \log(n) \Big) \nonumber \\ 
& \leq Pr\Big(d({H};{\hat{\boldsymbol{\sigma}}}) <   \frac{t+1}{\epsilon} \log (n) \Big) + o_n(1) 
\end{align}
Let $\Psi(H;\boldsymbol{\hat \sigma})$ show the set of cross-cluster hyperedges in hypergraph $H$ w.r.t. labeling $\boldsymbol{\hat \sigma}$. Define:
\begin{align*}
  \Tilde{d} &\triangleq \min_{\boldsymbol{\hat \sigma}\neq \hat{\boldsymbol{\sigma}}_{\text{ML}}}  |\Psi(H;\boldsymbol{\hat \sigma})| - |\Psi(H;\hat{\boldsymbol{\sigma}}_{\text{ML}})|\\
   \boldsymbol{\sigma}_{\min} &\triangleq \arg \min_{\boldsymbol{\hat \sigma}\neq \hat{\boldsymbol{\sigma}}_{\text{ML}}}  |\Psi(H;\boldsymbol{\hat \sigma})| 
\end{align*}
By the law of total probability and using Lemma~\ref{lem:dhq_lb} we get:
\begin{align}
&Pr \big( d(H;\hat{\boldsymbol{\sigma}}) < \frac{t+1}{\epsilon} \log(n) \big) \nonumber\\
&\leq  Pr \big(\Tilde{d} < {\frac{t+1}{\epsilon} \log(n)} \big) \nonumber \\
& = Pr \big( \Tilde{d} < {\frac{t+1}{\epsilon} \log(n)} {\big |} \hat{\boldsymbol{\sigma}}_{\text{ML}} = \boldsymbol{\sigma}^\ast \big) Pr (\hat{\boldsymbol{\sigma}}_{\text{ML}} = \boldsymbol{\sigma}^\ast)  \nonumber\\
& \quad +  Pr \big(\Tilde{d} < {\frac{t+1}{\epsilon} \log(n)} {\big |} \hat{\boldsymbol{\sigma}}_{\text{ML}} \neq \boldsymbol{\sigma}^\ast \big) \big( 1- Pr(\hat{\boldsymbol{\sigma}}_{\text{ML}} =\boldsymbol{\sigma}^\ast) \big) \nonumber \\ 
& \leq Pr \big( \Tilde{d} < {\frac{t+1}{\epsilon} \log(n)}\big) + o_n(1)
\label{eq: prb_stab_lower}
\end{align}
where the inequality holds when:
\begin{align}
 \label{eq:exact_hyp}
  (\sqrt{a} - \sqrt{b})^2 > 2^{h-1} \Longrightarrow Pr (\hat{\boldsymbol{\sigma}}_{\text{ML}} = \boldsymbol{\sigma}^\ast) = 1-o_n(1)
\end{align}

Let $\boldsymbol{\hat \sigma} \triangleq (A, B)$ where:
\begin{align*}
   &A \triangleq \{ v \in \mathcal{V}: \quad{ \hat \sigma}_v = + \}\\
   &B \triangleq \{ v \in \mathcal{V}: \quad { \hat \sigma}_v = - \} 
\end{align*}
Then, $\boldsymbol{\sigma}^\ast = (A^\ast, B^\ast)$ and $ \boldsymbol{\sigma}_{\min} =(A_{min},B_{min})$. Assume $S_{1} \subseteq A^\ast$ and $S_{2} \subseteq B^\ast$ are the set of misclassified nodes in $A_{min}$ and $ B_{min}$ with respect to $A^\ast$ and $B^\ast$, respectively. The same-size communities assumption implies $|S_{1}| = |S_{2}|$. Then: 
\begin{align}
\label{eq:binom_s_dist}
&Pr\Big( |\Psi(H;\boldsymbol{\sigma}_{\min})| - |\Psi(H;\boldsymbol{\sigma}^\ast)| \leq {\frac{t+1}{\epsilon} \log(n)} {\Big |} |S_{1}| = s \Big) \nonumber \\
& = Pr\Big( X_{1}^{(m)} - X_{2}^{(m)} < \frac{t+1}{\epsilon} \log(n) \Big)
\end{align}
where ${X}_{1}^{(m)} \sim \mathrm{Binom}(m, p)$, ${X}_{2}^{(m)} \sim \mathrm{Binom}(m, q)$, and:
\begin{align}
 \label{eq:m}
 m = 2 \sum\limits_{i=1}^{\min(h-1,s)} \binom{s}{i} \binom{n-s}{h-i}
\end{align}
Lemma~\ref{lem:binom_dist_privacy} gives:

\begin{align}
    &Pr\bigg( X_{1}^{(m)} - X_{2}^{(m)} < \frac{(t+1) \log(n)}{\epsilon} \bigg) \nonumber \\
    &\qquad \qquad \qquad \leq   n^{-  \Big( \frac{ 4s (1-\frac{s}{n})}{2^{h-1}} \mu-\frac{t+1}{2\epsilon} \log{( \frac{a}{b})} \Big)}
\end{align}
where:
\begin{align}
\label{mu_def_proof}
    \mu = a+b- 2\sqrt{ab} \times \sqrt{ \frac{(t+1)^2}{16ab\epsilon^2}  \big(\frac {h}{h-1} \big)^{2h-2}+1}
\end{align}
Applying the union bound for  $\mu > 2^{h-1} $ results in:
\begin{align}
\label{eq:union_error}
&Pr \Big( |\Psi(H;\boldsymbol{\sigma}_{\min})| - |\Psi(H;\boldsymbol{\sigma}^\ast) | < \frac{t+1}{\epsilon} \log(n) \Big) \nonumber \\
&\leq \sum\limits_{s=1}^{n/2} {\binom{n}{s}}^{2} \times n^ {-  \Big[\frac{ 4s (1-\frac{s}{n})}{2^{h-1}} \mu-\frac{t+1}{2\epsilon} \log{( \frac{a}{b})} \Big]} \nonumber \\ 
& \overset{(a)}{\leq} \sum\limits_{s=1}^{n/2} \bigg( \frac{n e}{s}\bigg)^{2s} \times n^ {-  \Big[ \frac{ 4s (1-\frac{s}{n})}{2^{h-1}} \mu-\frac{t+1}{2\epsilon} \log{( \frac{a}{b})} \Big]} \nonumber\\
&= \sum\limits_{s=1}^{n/2} e^{2s \log ( \frac{n e}{s}) } \times e^ {-\log(n)  \Big[ \frac{ 4s (1-\frac{s}{n})}{2^{h-1}} \mu-\frac{t+1}{2\epsilon} \log{( \frac{a}{b})} \Big]} \nonumber \\
& \overset{(b)}{<} n^{-2+\frac{t+1}{2\epsilon} \log{( \frac{a}{b})}} \times \sum\limits_{s=1}^{n/2}  e^ {-  2s \big[ \log(s)-\frac{2s}{n} \log(n) -1 \big] } \nonumber \\
& \overset{(c)}{=} n^{-2 + \frac{t+1}{2\epsilon} \log{( \frac{a}{b})}} \times \sum\limits_{s=1}^{n/2}   e^ {-  \frac{2s}{3} (\log(s) - 3) } \nonumber \\
& \overset{(d)}{=} n^{-2 + \frac{t+1}{2\epsilon} \log{( \frac{a}{b})}} \times O_n(1) \nonumber \\
&\overset{(e)}{=} o_n(1)
\end{align}
where (a) is due to Stirling's approximation, (b) holds since $\mu > 2^{h-1}$ and $e^{-2} \geq e^{-2s}$ for all $s \in [1, n/2]$, (c) derived from by $ \log(2s) - \frac{2s}{n} \log(n) \geq \frac{1}{3} \log(2s)$ for sufficiently large $n$ and any $ k \in [1, n/4]$, and (d) is the direct result of $\sum\limits_{s=1}^{n/2}   e^ {-  \frac{2s}{3} (\log(s) - 3) } = O_n(1)$ according to a result of~\cite{abbe2015exact}. In~\eqref{eq:union_error}, step (e) is obtained by assuming:
\[ \frac{t+1}{2\epsilon} \log (\frac{a}{b}) \leq 1 \]

\section{Proof of Theorem~\ref{thm:rr_mechanism}}
\label{prf:them:rr_mechanism}
    First, we find the in-cluster and cross-cluster hyperedge probabilities in $\tilde H$. Recall ~\eqref{eq:inc_crc_prb} and~\eqref{mu_rr}, then:
\[
\begin{cases} \Tilde{p} = (1-\nu)p+ \nu (1-p) = \frac{1 + p (e^\epsilon - 1)}{e^\epsilon + 1}, &\text{in-cluster} \\
\\
 \Tilde{q} = (1-\nu)q+\nu (1-q) =  \frac{1 + q (e^\epsilon - 1)}{e^\epsilon + 1}, & \text{cross-cluster} \end{cases}
\]
Now, we have three different regimes to study based on the privacy budget $\epsilon$. 
\begin{itemize}
    \item If $\epsilon = o_n \big( \log (n) \big)  $, then $\Tilde{p} = \Tilde{q}$ asymptotically. This means no recovery.
    \item If $\epsilon = O_n \big( \log (n) \big) $, then $\Tilde{p} = p$ and $\Tilde{q} = q$ asymptotically. This implies that the recovery bounds for private and non-private cases are identical.
    \item If $\epsilon = (h-1) \log(n)$ such that $e^{-\epsilon} = \frac{\lambda \log(n)}{\binom{n-1}{h-1}}$, then:
    \begin{align}
    \begin{cases} \Tilde{p} &=  (a + \lambda ) \frac{\log (n)}{\binom{n-1}{h-1}} \\
    \Tilde{q} &= (b + \lambda ) \frac{\log (n)}{\binom{n-1}{h-1}} \end{cases}
    \end{align}
\end{itemize}
where $\lambda = \Omega_n(1)$. Then, $(\sqrt{a+\lambda}-\sqrt{b+\lambda}\,)^2 > 2^{h-1}$ is the recovery bound. 

\section{Proof of Theorem~\ref{thm:bayes_samp}}
\label{prf:thm:bayes_samp}
For a hypergraph $H$, a perturbed hypergraph $\tilde H$ is considered where  $|\mathcal{W}_H \Delta \mathcal{W}_{{\tilde H}}| = 1$, and the difference is the hyperedge $w$. Then:
\begin{align}
\label{eq:tildH}
    Pr({\tilde H}) =& \sum\limits_{\boldsymbol{\sigma'}: w\in \Psi(H;\boldsymbol{\sigma'})} Pr(\boldsymbol{\sigma'}) Pr({H}|\boldsymbol{\sigma'})  \frac{Pr({\tilde H}|\boldsymbol{\sigma'})}{Pr({H}|\boldsymbol{\sigma'})} \nonumber \\
    &+ \sum\limits_{\boldsymbol{\sigma'}: w\notin \Psi(H;\boldsymbol{\sigma'})}Pr(\boldsymbol{\sigma'})Pr({H}|\boldsymbol{\sigma'}) \frac{Pr({\tilde H}|\boldsymbol{\sigma'})}{Pr({H}|\boldsymbol{\sigma'})} \nonumber \\
    =& \sum\limits_{\boldsymbol{\sigma'}: w\in \Psi(H;\boldsymbol{\sigma'})} Pr(\boldsymbol{\sigma'}) Pr({H}|\boldsymbol{\sigma'})  \frac{1-q}{q} \nonumber \\
    &+ \sum\limits_{\boldsymbol{\sigma'}: w\notin \Psi(H;\boldsymbol{\sigma'})}Pr(\boldsymbol{\sigma'})Pr({H}|\boldsymbol{\sigma'}) \frac{1-p}{p} \nonumber \\
    \leq& \;\frac{1-q}{q} \quad Pr(H)
\end{align}
where the last inequality holds because $\frac{1-q}{q} \geq \frac{1-p}{p}$. Similarly:
\begin{align}
\label{eq:H}
    Pr({ H}) 
    =& \sum\limits_{\boldsymbol{\sigma'}: w\in \Psi(H;\boldsymbol{\sigma'})} Pr(\boldsymbol{\sigma'}) Pr({\tilde H}|\boldsymbol{\sigma'})  \frac{q}{1-q} \nonumber \\
    &+ \sum\limits_{\boldsymbol{\sigma'}: w\notin \Psi(H;\boldsymbol{\sigma'})} Pr(\boldsymbol{\sigma'}) Pr({\tilde H}|\boldsymbol{\sigma'}) \frac{p}{1-p} \nonumber \\
    \leq& \;\frac{p}{1-p} \quad Pr(\tilde H)
\end{align}

Now, we begin by considering the case when the graph $H$ is perturbed via an in-cluster hyperedge $w$. Then:
\begin{align}
\label{eq:H_tildH_inc}
\frac{Pr(\boldsymbol{\sigma}| { H})}{Pr(\boldsymbol{\sigma}|{\tilde H})} 
&= \frac{Pr({ H}|\boldsymbol{\sigma})}{Pr( {\tilde H}|\boldsymbol{\sigma})} \frac{Pr( {\tilde H})}{Pr( H)} \nonumber\\
& = \frac{p}{1-p} \quad \frac{Pr( { \tilde H})}{Pr( H)} \nonumber \\
& \overset{(a)}{\leq} \frac{p}{1-p} \quad \frac{1-q}{q} \nonumber\\
& =  e^{\log{ \left( \frac{p (1-q)}{q(1-p)} \right)}}
\end{align}
where (a) holds due to~\eqref{eq:tildH}. In addition:
\begin{align}
\label{eq:tildH_H_inc}
\frac{Pr(\boldsymbol{\sigma}| {\tilde H})}{Pr(\boldsymbol{\sigma}|{H})} 
& = \frac{1-p}{p} \quad \frac{Pr( { H})}{Pr( {\tilde H})} \nonumber \\
& \leq 1
\end{align}
where the last inequality comes from~\eqref{eq:H}.


Now similarly, if $w$ is a cross-cluster hyperedge:
\begin{align}
\label{eq:tildeH_H_crc}
\frac{Pr(\boldsymbol{\sigma}|{\tilde{H}})}{Pr(\boldsymbol{\sigma}|{H})} 
& = \frac{1-q}{q} \quad \frac{Pr( { H})}{Pr({\tilde H})} \nonumber \\ 
&\overset{(a)}{\leq} \frac{p(1-q)}{q(1-p)} \nonumber\\
&=e^{\log{ \left( \frac{p (1-q)}{q(1-p)} \right)}}
\end{align}
where (a) holds due to~\eqref{eq:H}. In addition, we can similarly show:
\begin{align}
\label{eq:H_tildeH_crc}
  \frac{Pr(\boldsymbol{\sigma}|{H})}{Pr(\boldsymbol{\sigma}|{\tilde{H}})}  \leq 1
\end{align}
Whether $w$ is in-cluster or cross-cluster hyperedge, \eqref{eq:H_tildH_inc}-\eqref{eq:H_tildeH_crc} show the Bayesian sampling mechanism satisfies $\epsilon$-hyperedge DP for all $\epsilon$ such that $\epsilon \geq  \log \bigg( \frac{p(1-q)}{q(1-p)} \bigg)$. Asymptotically, this gives $\epsilon \geq  \log (\frac{a}{b}) = \epsilon_{0}$.

Now, we analyze the error probability of the Bayesian mechanism. For a fixed hypergraph $H$, our goal is to show: 
 \begin{align}   \frac{Pr(\hat{\boldsymbol{\sigma}}_{\text{Bayes}} \neq \boldsymbol{\sigma}^\ast)}{Pr(\boldsymbol{\sigma}^\ast|{H})} & = \frac{\sum\limits_{\boldsymbol{\sigma} \neq \boldsymbol{\sigma}^\ast} Pr(\boldsymbol{\sigma} | {H})}{Pr(\boldsymbol{\sigma}^\ast | {H})} \nonumber \\ 
     & =  \frac{\sum\limits_{\boldsymbol{\sigma} \neq \boldsymbol{\sigma}^\ast} Pr({H} | \boldsymbol{\sigma} )}{Pr({H} | \boldsymbol{\sigma}^\ast)} \nonumber\\ 
     & \leq o_n(1) 
\end{align}
which is equivalent to $Pr(\hat{\boldsymbol{\sigma}}_{\text{Bayes}} \neq \boldsymbol{\sigma}^\ast) \leq o_n(1)$. Recall $S_{1}$ and $S_{2}$ as the misclassified labels and both of size $s$, where $s \in [1, \frac{n}{2}]$, and  ${X}_{1}^{(m)} \sim \mathrm{Binom}(m, p)$, ${X}_{2}^{(m)} \sim \mathrm{Binom}(m, q)$ such that:
\begin{align*}
 m &= 2 \sum\limits_{i=1}^{\min(h-1,s)} \binom{s}{i} \binom{n-s}{h-i}
\end{align*}
Therefore, applying a similar error analysis executed in the stability-based methods: 
\begin{align}
\label{eq: Bayes_error_sampling}
&\frac{Pr(\hat{\boldsymbol{\sigma}}_{\text{Bayes}} \neq \boldsymbol{\sigma}^\ast | {H})}{Pr(\boldsymbol{\sigma}^\ast|{H})}  \leq \sum\limits_{s=1}^{\frac{n}{2}} { \binom{n}{s}}^{2}  \times e^{\frac{a-b}{a} ( X_2^{(m)} - X_1^{(m)})} \nonumber \\
&\overset{(a)}{\leq} \sum\limits_{s=1}^{\frac{n}{2}} \bigg( \frac{ne}{s}\bigg)^{2s}  \times e^{\frac{a-b}{a} ( X_2^{(m)} - X_1^{(m)})} \nonumber \\
& = \sum\limits_{s=1}^{\frac{n}{2}} e^{- 2s( \log(s)  - 1  ) } \times n^{- \bigg( \frac{ m ( a-b )^2}{a\binom{n-1}{h-1}}  -2s\bigg) }  \nonumber \\
& \overset{(b)}{\leq}  o_n(1)
\end{align}
where (a) follows from Stirling's approximation, and (b) holds if:
\begin{align}
\label{eq: suff_expo}
    \frac{( a-b )^2}{a} \frac{m}{\binom{n-1}{h-1}} & \geq 2s\\
\end{align}
Based on Lemma~\ref{lem:m&n_m}:
\begin{align*}
    \frac{( a-b )^2}{a} \frac{m}{\binom{n-1}{h-1}} \geq 2s \frac{( a-b )^2}{a} (1-\frac{s}{n})^{h-1}
\end{align*}
therefore, \eqref{eq: suff_expo} automatically holds if:
\begin{align*}
    ( a-b ) \frac{a-b}{a} (1-\frac{s}{n})^{h-1} & \geq 1 \\
    & \geq (1-\frac{s}{n})^{1-h}\\
    & \geq 2^{h-1}
\end{align*}
Since $ \frac{b}{a} = e^{-\epsilon_0}$, this completes the proof.

\section{Proof of Theorem~\ref{thm:exponential_samp}}
\label{prf:thm:exponential_samp}
Recall $S_{1}$ and $S_{2}$ as the misclassified labels and both of size $s$, where $s \in [1, \frac{n}{2}]$, and  ${X}_{1}^{(m)} \sim \mathrm{Binom}(m, p)$, ${X}_{2}^{(m)} \sim \mathrm{Binom}(m, q)$ such that:
\begin{align*}
 m &= 2 \sum\limits_{i=1}^{\min(h-1,s)} \binom{s}{i} \binom{n-s}{h-i}
\end{align*}
Now, we analyze the error probability of $\mathcal{M}_{\operatorname{Expo.}}(H)$. Similar to the Bayesian mechanism: 
  \begin{align}
  \label{eq:expo_in_ratio}
    \frac{e^{-\epsilon |\Psi(H;\boldsymbol{\sigma})|}}{e^{-\epsilon | \Psi(H;\boldsymbol{\sigma}^\ast) |}} &= \frac{e^{-\epsilon \big[| \Psi(H;\boldsymbol{\sigma}^\ast) | + X_1^{(m)} - X_2^{(m)}\big]}}{e^{-\epsilon | \Psi(H;\boldsymbol{\sigma}^\ast) | }} \nonumber \\
      & = e^{-\epsilon ( X_1^{(m)} - X_2^{(m)})} 
  \end{align}
Therefore:
\begin{align}
\label{eq:expo_ratio}
     &\frac{Pr(\hat{\boldsymbol{\sigma}}_{\text{Expo.}} \neq \boldsymbol{\sigma}^\ast | {H})}{Pr(\boldsymbol{\sigma}^\ast|{H})}   \leq  \sum_{s=1}^{n/2} {\binom{n}{s}}^{2}  e^{-\epsilon ( X_1^{(m)} - X_2^{(m)})} \nonumber \\
      & \leq \sum\limits_{s=1}^{\frac{n}{2}}  e^{- 2s ( \log(s)  - 1  ) } \times n^{- \bigg( \frac{ m ( a-b ) \epsilon }{\binom{n-1}{h-1}}  -2s \bigg) }   \nonumber \\
      & \leq \sum\limits_{s=1}^{\frac{n}{2}}  e^{- 2s ( \log(s)  - 1  ) } \times n^{- \big( \frac{  (a-b ) \epsilon }{2^{h-1}}  -2s \big) }   \nonumber \\
      & \leq o_n(1)
\end{align}
where the last inequality holds if $a - b > \frac{2^{h-1}}{\epsilon}$, and this completes the proof.
\section{Auxiliary Lemmas}
\begin{lemma}
\label{lem:dhq_lb} 
For any random $h$-uniform hypergraph $H$ and any estimator $\hat{\boldsymbol{\sigma}}(H)$, the distance to instability satisfies:
\begin{align}
 d({H},{\hat{\boldsymbol{\sigma}}}) \geq  \min_{\boldsymbol{\hat \sigma}\neq \hat{\boldsymbol{\sigma}}_{\text{ML}}} |\Psi(H;\boldsymbol{\hat \sigma})| - |\Psi(H;\hat{\boldsymbol{\sigma}}_{\text{ML}})|
 \label{eq:DTI-bound}
\end{align}
where $\boldsymbol{\hat \sigma}_{\text{ML}}$ is the community vector recovered by the maximum likelihood estimator on hypergraph $H$.
\end{lemma}
\begin{proof}
Based on Definition~\ref{def:dist2inst}, $ d({H},\boldsymbol{\hat \sigma}_{\text{ML}})$ indicates the minimum distance hypergraph $H$ and a neighboring hypergraph $\tilde H$ have such that $\boldsymbol{\hat \sigma}_\text{ML}(H) \neq \boldsymbol{\hat \sigma}_\text{ML}(\tilde{H})$. Accordingly:
\begin{align*}
    d({H},\boldsymbol{\hat \sigma}_{\text{ML}}) = |\mathcal{W}_H \Delta \mathcal{W}_{\tilde{H}}|
\end{align*}
Recall $\Psi(\cdot,\cdot)$ was defined  as the set of   cross-cluster hyperedges for a given hypergraph and set of communities. By contraposition, assume Eq.~\eqref{eq:DTI-bound} is violated, i.e., for some hypergraph $H$ and its maximum likelihood community estimate $\boldsymbol{\hat\sigma}_{\text{ML}}$:
\begin{align*}
 d({H},\boldsymbol{\hat \sigma}_{\text{ML}}) <\min_{\boldsymbol{\hat \sigma}\neq \hat{\boldsymbol{\sigma}}_{\text{ML}}} | \Psi(H;\boldsymbol{\hat \sigma}) | - | \Psi(H;\hat{\boldsymbol{\sigma}}_{\text{ML}}) |
\end{align*}
It follows that:
\begin{align}
| \Psi(\tilde{H};\boldsymbol{\hat \sigma}_\text{ML}(H))| 
&\leq | \Psi(H;\boldsymbol{\hat \sigma}_\text{ML}(H)) | + |\mathcal{W}_H \Delta \mathcal{W}_{\tilde{H}}| \nonumber \\
& = | \Psi(H;\boldsymbol{\hat \sigma}_\text{ML}(H)) | + d({H},\boldsymbol{\hat \sigma}_{\text{ML}}) \nonumber \\
& < \min_{\boldsymbol{\hat \sigma}\neq \hat{\boldsymbol{\sigma}}_{\text{ML}}} | \Psi(H;\boldsymbol{\hat \sigma} (H) ) |
 \label{eq:k_stable_lower_bound}
\end{align}
But then, the last inequality implies $\boldsymbol{\hat \sigma}_\text{ML}(\tilde{H}) = \boldsymbol{\hat \sigma}_\text{ML}(H)$, which contradicts the assumption of the lemma. This completes the proof.
\end{proof}
\begin{lemma}
    \label{lem:binom_dist_privacy}
    For $m = 2 \sum\limits_{i=1}^{\min(h-1,s)} \binom{s}{i} \binom{n-s}{h-i}$, where $s \in [n/2]$, and $t$ and $\epsilon$ are  positive scalars:
    \begin{align*}
        &Pr \Big(\mathrm{Binom}(m, p) - \mathrm{Binom}(m, q) < \frac{t+1}{\epsilon} \log(n) \Big) \nonumber \\
        & \qquad \leq n^{-  \Big[ \frac{ 4s (1-\frac{s}{n})}{2^{h-1}} \mu-\frac{t+1}{2\epsilon} \log{( \frac{a}{b})} \Big]}
    \end{align*}
where $p = \frac{a \log(n)}{\binom{n-1}{h-1}}$, $q=\frac{b \log(n)}{\binom{n-1}{h-1}}$, for  $a \geq b >0$, and:
\begin{align*}
    \mu = a+b- 2\sqrt{ab} \times \sqrt{ \frac{(t+1)^2}{16ab\epsilon^2}  \big(\frac {h}{h-1} \big)^{2h-2}+1}
\end{align*}

\end{lemma}
\begin{proof}
\label{plem:binom_dist_privacy}
Let ${X}_{1}^{(m)} \sim \mathrm{Binom}(m, p)$, ${X}_{2}^{(m)} \sim \mathrm{Binom}(m, q)$, and $\beta = \frac{t+1}{\epsilon}$. Applying Chernoff's bound leads to:
\begin{align}
\label{eq:chernoff}
& Pr\big( X_{1}^{(m)} - X_{2}^{(m)}< \beta \log(n) \big) \nonumber \\
&  \leq \min_{\lambda > 0} e^{\lambda \beta \log(n)} \times \mathbb{E} \Big[e^{-\lambda( {X}_{1}^{(m)} - {X}_{2}^{(m)} ) }\Big] \nonumber \\ 
& = \min_{\lambda > 0} e^{\lambda \beta \log(n) } \times (1- p (1- e^{-\lambda}))^{m} \times  (1- q (1- e^{\lambda}))^{m} \nonumber \\
& \leq \min_{\lambda > 0} e^{\lambda \beta \log(n) } \times e^{-mp(1- e^{-\lambda})-mq(1- e^{-\lambda})} \nonumber \\
&=  \min_{\lambda > 0} \quad n^{-\theta(\lambda)}
\end{align}
where
\begin{align}
\theta(\lambda) = \frac{m}{\binom{n-1}{h-1}} (a+b-ae^{-\lambda}-be^{\lambda}) -\lambda \beta
\end{align}
Let $\gamma = \sqrt{\beta^2+ \frac{4m^2ab}{\binom{n-1}{h-1}^2}}$, then:
\begin{align}
\label{eq:theta_m}
&\theta^* =\max_{\lambda} (\theta) \nonumber\\
&= \frac{m}{\binom{n-1}{h-1}} \Bigg( a + b - 2\sqrt{ab} \times \sqrt{\frac{\beta^2 \binom{n-1}{h-1}^2}{4m^2ab}+1} \;\Bigg) \nonumber \\
& \quad -\frac{\beta}{2} \log{( \frac{a}{b} )} + \frac{\beta}{2} \log{\Big(\frac{\gamma + \beta}{\gamma - \beta}\Big)} \nonumber \\
&\geq \frac{m}{\binom{n-1}{h-1}} \Bigg( a + b - 2\sqrt{ab} \times \sqrt{\frac{\beta^2 \binom{n-1}{h-1}^2}{4m^2ab}+1} \;\Bigg) 
-\frac{\beta}{2} \log{\Big( \frac{a}{b}\Big)}
\end{align}
where the last inequality is directly follows from $\frac{\beta}{2} \log{\Big(\frac{\gamma + \beta}{\gamma - \beta}\Big)} \geq 0$. We now wish to remove the dependence on $m$ in the lower bound~\eqref{eq:theta_m}. By Lemma~\ref{lem:m&n_m}, for sufficiently large $n$:
\begin{align}
 \label{eq:n_m_general2}
 m \geq  2s(1-\frac{s}{n})^{h-1} \binom{n-1}{h-1}
\end{align}
Define $\kappa \triangleq \frac{n}{s}$. Obviously, $\kappa \geq 2$. We have:
\begin{align}
\label{eq:theta_max2}
    \theta^* &\geq  2s(1-\frac{s}{n})^{h-1} \Bigg[ a+b- \sqrt{ \frac{ \beta^2 \kappa^2}{4n^2} \big(\frac {\kappa}{\kappa-1} \big)^{2h-2} +4ab} \; \Bigg] \nonumber \\
    & \qquad - \frac{\beta}{2} \log \big( \frac{a}{b} \big)
\end{align}
Let $g(\kappa) \triangleq \kappa^2 \big(\frac {\kappa}{\kappa-1} \big)^{2h-2}$, a concave function of $\kappa$. Then: 
\begin{align*}
    \kappa^* &=\underset{\kappa}{\arg \max} \, g(\kappa) = h \\
    g(\kappa^*) &= {h^2} \big(\frac {h}{h-1} \big)^{2h-2}
\end{align*}
therefore, for sufficiently large $n$:
\begin{align}
\label{eq:theta_max3}
    \theta^* &\geq 2s (1-\frac{s}{n})^{h-1} \Bigg[ a+b- 2 \sqrt{ \frac{\beta^2}{16n^2}  \big(\frac {h}{h-1} \big)^{2h-2}+ab} \;\Bigg]  \nonumber \\
    &\quad - \frac{\beta}{2} \log \big( \frac{a}{b} \big) \nonumber \\
    &\geq \frac{ 4s (1-\frac{s}{n})}{2^{h-1}} \Bigg[ a+b- 2\sqrt{ab} \times \sqrt{ \frac{\beta^2}{16ab}  \big(\frac {h}{h-1} \big)^{2h-2}+1} \;\Bigg] \nonumber \\
    &\quad- \frac{\beta}{2} \log \big( \frac{a}{b} \big)
\end{align}
due to~$(1-\frac{s}{n})^{h-2} \geq \frac{1}{2^{h-2}}$ for $s \in [1, n/2]$.
Substituting $\theta^*$ back into Eq.~\eqref{eq:chernoff}:
\begin{align}
     Pr\big( X_{1}^{(m)} - X_{2}^{(m)}< \beta \log(n) \big) \leq n^{-  \Big[ \frac{ 4s (1-\frac{s}{n})}{2^{h-1}} \mu-\frac{\beta}{2} \log{( \frac{a}{b})} \Big]}
\end{align}
This completes the proof.
\end{proof}

\begin{lemma}
    \label{lem:m&n_m}
        If $m = 2 \sum\limits_{i=1}^{\min(h-1,s)} \binom{s}{i} \binom{n-s}{h-i}$, where $s \in [n/2]$, and $h\geq2$ is an integer, then for sufficiently large $n$:
    \begin{align*}
        m \geq  2s(1-\frac{s}{n})^{h-1} {\binom{n-1}{h-1}}
    \end{align*}
\end{lemma}

\begin{proof}
\begin{align*}
 &\frac{m}{{\binom{n-1}{h-1}}} = \frac{2 \sum\limits_{i=1}^{\min(h-1,s)}  {\binom{s}{i}} \binom{n-s}{h-i}}{\binom{n-1}{h-1}}\\ 
 =& 2s \Bigg[ \Big[ \frac{(1- \frac{s}{n}) \cdots (1 - \frac{s+h-2}{n})} {(1 - \frac{1}{n}) \cdots (1 -\frac{h -1}{n}) } \Big] + \\
 &\quad \frac{s}{n} \frac{h-1}{ 2!}\Big[ \frac{(1-\frac{1}{s})(1- \frac{s}{n}) \cdots (1- \frac{s+h-3}{n})} {(1 - \frac{1}{n}) \cdots (1 -\frac{h -1}{n}) } \Big] + \cdots  \nonumber \\
 &\quad + (\frac{s}{n})^{h-2} \frac{(h-1)!}{(h-1)!}\Big[ \frac{(1 - \frac{s}{n}) (1 - \frac{1}{s}) \cdots (1 -\frac{h -2}{s})} {(1 - \frac{1}{n}) \cdots (1 -\frac{h -1}{n}) } \Big] \Bigg] \nonumber \\
 =& 2s \Bigg[ (1- \frac{s}{n})^{h-1}  + \frac{s}{n}  \frac{h-1}{ 2!}  (1-\frac{s}{n})^{h-2}+ \cdots  \\
 &\quad + (\frac{s}{n})^{h-2}  (1-\frac{s}{n}) \Bigg] \big(1-o_n(1)\big)\nonumber \\
 &\overset{(a)}{=} 2s (1- \frac{s}{n})^{h-1} \Big[   1+  \frac{h-1}{ 2!} + \cdots +\frac{(h-1)!}{ (h-1)!} \Big] \big(1-o_n(1)\big)\nonumber \\
&\geq  2s(1-\frac{s}{n})^{h-1} 
\end{align*}
where (a) holds since $1-\frac{s}{n} \geq \frac{s}{n}$. Then from the last inequality, we have:
\begin{equation*}
  m \geq  2s(1-\frac{s}{n})^{h-1} {\binom{n-1}{h-1}}
\end{equation*}
\end{proof}
\ifCLASSOPTIONcaptionsoff
  \newpage
\fi
\bibliographystyle{IEEEtran}
\bibliography{Ref.bib}

\begin{thebibliography}{10}
\providecommand{\url}[1]{#1}
\csname url@samestyle\endcsname
\providecommand{\newblock}{\relax}
\providecommand{\bibinfo}[2]{#2}
\providecommand{\BIBentrySTDinterwordspacing}{\spaceskip=0pt\relax}
\providecommand{\BIBentryALTinterwordstretchfactor}{4}
\providecommand{\BIBentryALTinterwordspacing}{\spaceskip=\fontdimen2\font plus
\BIBentryALTinterwordstretchfactor\fontdimen3\font minus \fontdimen4\font\relax}
\providecommand{\BIBforeignlanguage}[2]{{%
\expandafter\ifx\csname l@#1\endcsname\relax
\typeout{** WARNING: IEEEtran.bst: No hyphenation pattern has been}%
\typeout{** loaded for the language `#1'. Using the pattern for}%
\typeout{** the default language instead.}%
\else
\language=\csname l@#1\endcsname
\fi
#2}}
\providecommand{\BIBdecl}{\relax}
\BIBdecl

\bibitem{8661646}
F.~Ahmed, A.~X. Liu, and R.~Jin, ``Publishing social network graph eigenspectrum with privacy guarantees,'' \emph{IEEE Transactions on Network Science and Engineering}, vol.~7, no.~2, pp. 892--906, 2020.

\bibitem{9521831}
T.~Gao and F.~Li, ``Differential private social network publication and persistent homology preservation,'' \emph{IEEE Transactions on Network Science and Engineering}, vol.~8, no.~4, pp. 3152--3166, 2021.

\bibitem{9669059}
H.~Gao and Y.~Wang, ``Algorithm-level confidentiality for average consensus on time-varying directed graphs,'' \emph{IEEE Transactions on Network Science and Engineering}, vol.~9, no.~2, pp. 918--931, 2022.

\bibitem{9964113}
M.~Zhang, J.~Zhou, G.~Zhang, L.~Cui, T.~Gao, and S.~Yu, ``Apdp: Attribute-based personalized differential privacy data publishing scheme for social networks,'' \emph{IEEE Transactions on Network Science and Engineering}, vol.~10, no.~2, pp. 922--933, 2023.

\bibitem{10049709}
R.~Jin, Y.~Huang, Z.~Zhang, and H.~Dai, ``On the privacy guarantees of gossip protocols in general networks,'' \emph{IEEE Transactions on Network Science and Engineering}, vol.~10, no.~6, pp. 3114--3130, 2023.

\bibitem{10232888}
Y.~Zhu, Q.~Song, and Y.~Luo, ``Differentially private top-$k$ flows estimation mechanism in network traffic,'' \emph{IEEE Transactions on Network Science and Engineering}, vol.~11, no.~3, pp. 2462--2472, 2024.

\bibitem{10304338}
C.~Hawkins, B.~Chen, K.~Yazdani, and M.~Hale, ``Node and edge differential privacy for graph laplacian spectra: Mechanisms and scaling laws,'' \emph{IEEE Transactions on Network Science and Engineering}, vol.~11, no.~2, pp. 1690--1701, 2024.

\bibitem{10877783}
M.~Li and D.~Xiao, ``Communication-efficient and utility-enhanced local differential privacy-based personalized federated compressed learning,'' \emph{IEEE Transactions on Network Science and Engineering}, vol.~12, no.~3, pp. 1776--1790, 2025.

\bibitem{11078900}
Y.~Sun, Z.~Liu, Y.~Xia, Z.~Guo, G.~Liu, L.~Li, and J.~Ma, ``Ppdr: A privacy-preserving dual reputation management scheme in vehicle platoon,'' \emph{IEEE Transactions on Dependable and Secure Computing}, pp. 1--18, 2025.

\bibitem{10947335}
Y.~Hou, C.~Tran, M.~Li, and W.-Y. Shin, ``A unified framework for exploratory learning-aided community detection under topological uncertainty,'' \emph{IEEE Transactions on Network Science and Engineering}, pp. 1--17, 2025.

\bibitem{9057414}
Z.~Guan, Z.~Lv, X.~Sun, L.~Wu, J.~Wu, X.~Du, and M.~Guizani, ``A differentially private big data nonparametric bayesian clustering algorithm in smart grid,'' \emph{IEEE Transactions on Network Science and Engineering}, vol.~7, no.~4, pp. 2631--2641, 2020.

\bibitem{10236973}
K.~Guo, D.~Chen, Q.~Huang, F.~Li, C.~Guo, D.~Wu, X.~Liu, and K.~Chen, ``Privacy-preserving multi-label propagation based on federated learning,'' \emph{IEEE Transactions on Network Science and Engineering}, vol.~11, no.~1, pp. 886--899, 2024.

\bibitem{dwork2006differential}
C.~Dwork, ``Differential privacy,'' in \emph{International colloquium on automata, languages, and programming}.\hskip 1em plus 0.5em minus 0.4em\relax Springer, 2006, pp. 1--12.

\bibitem{dwork2008differential}
------, ``Differential privacy: A survey of results,'' in \emph{International conference on theory and applications of models of computation}.\hskip 1em plus 0.5em minus 0.4em\relax Springer, 2008, pp. 1--19.

\bibitem{dwork2014algorithmic}
C.~Dwork, A.~Roth \emph{et~al.}, ``The algorithmic foundations of differential privacy.'' \emph{Foundations and Trends in Theoretical Computer Science}, vol.~9, no. 3-4, pp. 211--407, 2014.

\bibitem{mulle2015privacy}
Y.~M{\"u}lle, C.~Clifton, and K.~B{\"o}hm, ``Privacy-integrated graph clustering through differential privacy.'' in \emph{EDBT/ICDT Workshops}, vol. 157, 2015.

\bibitem{mohamed2022differentially}
M.~S. Mohamed, D.~Nguyen, A.~Vullikanti, and R.~Tandon, ``Differentially private community detection for stochastic block models,'' in \emph{International Conference on Machine Learning}.\hskip 1em plus 0.5em minus 0.4em\relax PMLR, 2022, pp. 15\,858--15\,894.

\bibitem{seif2023differentially}
M.~Seif, A.~J. Goldsmith, and H.~V. Poor, ``Differentially private community detection over stochastic block models with graph sketching,'' in \emph{2023 57th Annual Conference on Information Sciences and Systems (CISS)}.\hskip 1em plus 0.5em minus 0.4em\relax IEEE, 2023, pp. 1--6.

\bibitem{holland1983stochastic}
P.~W. Holland, K.~B. Laskey, and S.~Leinhardt, ``Stochastic blockmodels: First steps,'' \emph{Social networks}, vol.~5, no.~2, pp. 109--137, 1983.

\bibitem{abbe2015exact}
E.~Abbe, A.~S. Bandeira, and G.~Hall, ``Exact recovery in the stochastic block model,'' \emph{IEEE Transactions on Information Theory}, vol.~62, no.~1, pp. 471--487, 2015.

\bibitem{yun2014accurate}
S.-Y. Yun and A.~Proutiere, ``Accurate community detection in the stochastic block model via spectral algorithms,'' \emph{arXiv preprint arXiv:1412.7335}, 2014.

\bibitem{chen2015spectral}
P.-Y. Chen and A.~O. Hero, ``Phase transitions in spectral community detection of large noisy networks,'' in \emph{2015 IEEE International Conference on Acoustics, Speech and Signal Processing (ICASSP)}, 2015, pp. 3402--3406.

\bibitem{mossel2015consistency}
E.~Mossel, J.~Neeman, and A.~Sly, ``Consistency thresholds for the planted bisection model,'' in \emph{Proceedings of the forty-seventh annual ACM symposium on Theory of computing}, 2015, pp. 69--75.

\bibitem{8793181}
L.~Su, W.~Wang, and Y.~Zhang, ``Strong consistency of spectral clustering for stochastic block models,'' \emph{IEEE Transactions on Information Theory}, vol.~66, no.~1, pp. 324--338, 2020.

\bibitem{yun2019optimal}
S.-Y. Yun and A.~Prouti{\`e}re, ``Optimal sampling and clustering in the stochastic block model,'' \emph{Advances in Neural Information Processing Systems}, vol.~32, 2019.

\bibitem{gangrade2019efficient}
A.~Gangrade, P.~Venkatesh, B.~Nazer, and V.~Saligrama, ``Efficient near-optimal testing of community changes in balanced stochastic block models,'' \emph{Advances in Neural Information Processing Systems}, vol.~32, 2019.

\bibitem{gaudio2022exact}
J.~Gaudio, M.~Z. Racz, and A.~Sridhar, ``Exact community recovery in correlated stochastic block models,'' in \emph{Conference on Learning Theory}.\hskip 1em plus 0.5em minus 0.4em\relax PMLR, 2022, pp. 2183--2241.

\bibitem{zhang2023fundamental}
A.~Y. Zhang, ``Fundamental limits of spectral clustering in stochastic block models,'' \emph{arXiv preprint arXiv:2301.09289}, 2023.

\bibitem{massoulie2014community}
L.~Massouli{\'e}, ``Community detection thresholds and the weak ramanujan property,'' in \emph{Proceedings of the forty-sixth annual ACM symposium on Theory of computing}, 2014, pp. 694--703.

\bibitem{hajek2016achieving}
B.~Hajek, Y.~Wu, and J.~Xu, ``Achieving exact cluster recovery threshold via semidefinite programming,'' \emph{IEEE Transactions on Information Theory}, vol.~62, no.~5, pp. 2788--2797, 2016.

\bibitem{hajek2016achieving_extensions}
------, ``Achieving exact cluster recovery threshold via semidefinite programming: Extensions,'' \emph{IEEE Transactions on Information Theory}, vol.~62, no.~10, pp. 5918--5937, 2016.

\bibitem{jalali2016exploiting}
A.~Jalali, Q.~Han, I.~Dumitriu, and M.~Fazel, ``Exploiting tradeoffs for exact recovery in heterogeneous stochastic block models,'' \emph{Advances in Neural Information Processing Systems}, vol.~29, 2016.

\bibitem{esm_2021_sdp}
M.~Esmaeili, H.~M. Saad, and A.~Nosratinia, ``Semidefinite programming for community detection with side information,'' \emph{IEEE Transactions on Network Science and Engineering}, vol.~8, no.~2, pp. 1957--1973, 2021.

\bibitem{yan2021covariate}
B.~Yan and P.~Sarkar, ``Covariate regularized community detection in sparse graphs,'' \emph{Journal of the American Statistical Association}, vol. 116, no. 534, pp. 734--745, 2021.

\bibitem{cbm_javad}
J.~Z. Moghaddam, M.~Esmaeili, and A.~Nosratinia, ``Exact recovery threshold in dynamic binary censored block model,'' in \emph{2022 IEEE International Symposium on Information Theory (ISIT)}, 2022, pp. 1088--1093.

\bibitem{javadTnse24}
J.~Z. Moghaddam and A.~Nosratinia, ``Community detection in dynamic networks: Exact recovery under two link evolution models,'' \emph{IEEE Transactions on Network Science and Engineering}, pp. 1--11, 2024.

\bibitem{sandon2017community}
C.~Sandon, ``Community detection in the stochastic block model: fundamental limits,'' Ph.D. dissertation, Princeton University, 2017.

\bibitem{Hajek_Wu_Xu_2018}
B.~Hajek, Y.~Wu, and J.~Xu, ``Recovering a hidden community beyond the kesten–stigum threshold in o(|e|log*|v|) time,'' \emph{Journal of Applied Probability}, vol.~55, no.~2, p. 325–352, 2018.

\bibitem{saad_side18}
H.~Saad and A.~Nosratinia, ``Community detection with side information: Exact recovery under the stochastic block model,'' \emph{IEEE Journal of Selected Topics in Signal Processing}, vol.~12, no.~5, pp. 944--958, 2018.

\bibitem{wu2021streaming}
Y.~Wu, J.~Tardos, M.~Bateni, A.~Linhares, F.~M. Goncalves~de Almeida, A.~Montanari, and A.~Norouzi-Fard, ``Streaming belief propagation for community detection,'' \emph{Advances in Neural Information Processing Systems}, vol.~34, pp. 26\,976--26\,988, 2021.

\bibitem{abbe2017community}
E.~Abbe, ``Community detection and stochastic block models: recent developments,'' \emph{The Journal of Machine Learning Research}, vol.~18, no.~1, pp. 6446--6531, 2017.

\bibitem{ning2023comprehensive}
S.~Ning, J.~Li, and Y.~Lu, ``A comprehensive review of community detection in graphs,'' \emph{arXiv preprint arXiv:2309.11798}, 2023.

\bibitem{ghoshdastidar2014consistency}
D.~Ghoshdastidar and A.~Dukkipati, ``Consistency of spectral partitioning of uniform hypergraphs under planted partition model,'' \emph{Advances in Neural Information Processing Systems}, vol.~27, 2014.

\bibitem{ghoshdastidar2017consistency}
D.~GHOSHDASTIDAR and A.~DUKKIPATI, ``Consistency of spectral hypergraph partitioning under planted partition model,'' \emph{The Annals of Statistics}, vol.~45, no.~1, pp. 289--315, 2017.

\bibitem{kim2017community}
C.~Kim, A.~S. Bandeira, and M.~X. Goemans, ``Community detection in hypergraphs, spiked tensor models, and sum-of-squares,'' in \emph{2017 International Conference on Sampling Theory and Applications (SampTA)}.\hskip 1em plus 0.5em minus 0.4em\relax IEEE, 2017, pp. 124--128.

\bibitem{lin2017fundamental}
C.-Y. Lin, I.~E. Chien, and I.-H. Wang, ``On the fundamental statistical limit of community detection in random hypergraphs,'' in \emph{2017 IEEE International Symposium on Information Theory (ISIT)}.\hskip 1em plus 0.5em minus 0.4em\relax IEEE, 2017, pp. 2178--2182.

\bibitem{kim2018stochastic}
C.~Kim, A.~S. Bandeira, and M.~X. Goemans, ``Stochastic block model for hypergraphs: Statistical limits and a semidefinite programming approach,'' \emph{arXiv preprint arXiv:1807.02884}, 2018.

\bibitem{chien2018community}
I.~Chien, C.-Y. Lin, and I.-H. Wang, ``Community detection in hypergraphs: Optimal statistical limit and efficient algorithms,'' in \emph{International Conference on Artificial Intelligence and Statistics}.\hskip 1em plus 0.5em minus 0.4em\relax PMLR, 2018, pp. 871--879.

\bibitem{ahn2018hypergraph}
K.~Ahn, K.~Lee, and C.~Suh, ``Hypergraph spectral clustering in the weighted stochastic block model,'' \emph{IEEE Journal of Selected Topics in Signal Processing}, vol.~12, no.~5, pp. 959--974, 2018.

\bibitem{ahn2019community}
------, ``Community recovery in hypergraphs,'' \emph{IEEE Transactions on Information Theory}, vol.~65, no.~10, pp. 6561--6579, 2019.

\bibitem{chien2019minimax}
I.~E. Chien, C.-Y. Lin, and I.-H. Wang, ``On the minimax misclassification ratio of hypergraph community detection,'' \emph{IEEE Transactions on Information Theory}, vol.~65, no.~12, pp. 8095--8118, 2019.

\bibitem{cole2020exact}
S.~Cole and Y.~Zhu, ``Exact recovery in the hypergraph stochastic block model: A spectral algorithm,'' \emph{Linear Algebra and its Applications}, vol. 593, pp. 45--73, 2020.

\bibitem{lee2020robust}
J.~Lee, D.~Kim, and H.~W. Chung, ``Robust hypergraph clustering via convex relaxation of truncated mle,'' \emph{IEEE Journal on Selected Areas in Information Theory}, vol.~1, no.~3, pp. 613--631, 2020.

\bibitem{zhang2022exact}
Q.~Zhang and V.~Y. Tan, ``Exact recovery in the general hypergraph stochastic block model,'' \emph{IEEE Transactions on Information Theory}, vol.~69, no.~1, pp. 453--471, 2022.

\bibitem{gaudio2023community}
J.~Gaudio and N.~Joshi, ``Community detection in the hypergraph sbm: Exact recovery given the similarity matrix,'' in \emph{The Thirty Sixth Annual Conference on Learning Theory}.\hskip 1em plus 0.5em minus 0.4em\relax PMLR, 2023, pp. 469--510.

\bibitem{wang2023projected}
J.~Wang, Y.-M. Pun, X.~Wang, P.~Wang, and A.~M.-C. So, ``Projected tensor power method for hypergraph community recovery,'' in \emph{International Conference on Machine Learning}.\hskip 1em plus 0.5em minus 0.4em\relax PMLR, 2023, pp. 36\,285--36\,307.

\bibitem{deng2023strong}
C.~Deng, X.-J. Xu, and S.~Ying, ``Strong consistency of spectral clustering for the sparse degree-corrected hypergraph stochastic block model,'' \emph{IEEE Transactions on Information Theory}, 2023.

\bibitem{nguyen2016detecting}
H.~H. Nguyen, A.~Imine, and M.~Rusinowitch, ``Detecting communities under differential privacy,'' in \emph{Proceedings of the 2016 ACM on Workshop on Privacy in the Electronic Society}, 2016, pp. 83--93.

\bibitem{qin2017generating}
Z.~Qin, T.~Yu, Y.~Yang, I.~Khalil, X.~Xiao, and K.~Ren, ``Generating synthetic decentralized social graphs with local differential privacy,'' in \emph{Proceedings of the 2017 ACM SIGSAC Conference on Computer and Communications Security (CCS)}, 2017, pp. 425--438.

\bibitem{8999786}
T.~Ji, C.~Luo, Y.~Guo, Q.~Wang, L.~Yu, and P.~Li, ``Community detection in online social networks: A differentially private and parsimonious approach,'' \emph{IEEE Transactions on Computational Social Systems}, vol.~7, no.~1, pp. 151--163, 2020.

\bibitem{hehir2021consistency}
J.~Hehir, A.~Slavkovic, and X.~Niu, ``Consistency of privacy-preserving spectral clustering under the stochastic block model,'' \emph{arXiv preprint arXiv:2105.12615}, 2021.

\bibitem{ZHOU2026103621}
W.~Zhou, Z.~Liu, A.~{Ul Haq}, Y.~Li, and Z.~L. Jiang, ``Differentially private matrix factorization with sub-linear convergence rate for personalized recommendation,'' \emph{Information Fusion}, vol. 126, p. 103621, 2026.

\bibitem{imola2021locally}
J.~Imola, T.~Murakami, and K.~Chaudhuri, ``Locally differentially private analysis of graph statistics,'' in \emph{30th USENIX security symposium (USENIX Security 21)}, 2021, pp. 983--1000.

\bibitem{ji2019differentially}
T.~Ji, C.~Luo, Y.~Guo, J.~Ji, W.~Liao, and P.~Li, ``Differentially private community detection in attributed social networks,'' in \emph{Asian Conference on Machine Learning}.\hskip 1em plus 0.5em minus 0.4em\relax PMLR, 2019, pp. 16--31.

\bibitem{ghoshdastidar2017uniform}
D.~Ghoshdastidar and A.~Dukkipati, ``Uniform hypergraph partitioning: Provable tensor methods and sampling techniques,'' \emph{Journal of Machine Learning Research}, vol.~18, no.~50, pp. 1--41, 2017.

\end{thebibliography}
\end{document}